\documentclass[journal]{IEEEtran}
\ifCLASSINFOpdf
\else
\fi
\usepackage{amsmath,graphicx}
\usepackage{mathrsfs}
\usepackage{amsfonts}
\usepackage{amsthm}
\usepackage{dsfont}
\theoremstyle{remark}
\newtheorem{definition}{Definition}

\newtheorem{theorem}{Theorem}[section]
\newtheorem{lemma}{Lemma}[section]
\newtheorem{corollary}{Corollary}[section]

\newtheorem{assumption}{Assumption}
\usepackage{color}
\usepackage{comment}
\usepackage{bm}

\usepackage{multicol}
\usepackage{cite}
\usepackage{url}

\newcommand{\R}{\ensuremath{\mathbb{R}}}

\def\diag{\text{diag}}
\def\0{\mathbf{0}}
\def\1{\mathds{1}}
\def\a{\mathbf{a}}
\def\A{\mathbf{A}}

\def\B{\mathbf{B}}
\def\C{\mathbf{C}}

\def\D{\mathbf{D}}
\def\e{\mathbf{e}}
\def\ep{\boldsymbol{\epsilon}}
\def\E{\mathcal{E}}
\def\Ep{\mathbb{E}}
\def\Ebf{\mathbf{E}}
\def\f{\mathbf{f}}

\def\Fbf{\mathbf{F}}

\def\G{\mathcal{G}}
\def\Gbf{\mathbf{G}}
\def\H{\mathbf{H}}
\def\h{\mathbf{h}}
\def\I{\mathbf{I}}
\def\J{\mathbf{J}}
\def\Kbf{\mathbf{K}}
\def\Lbf{\mathbf{L}}

\def\Mbf{\mathbf{M}}

\def\Obf{\mathcal{O}}

\def\P{\mathbf{P}}
\def\rbf{\mathbf{r}}
\def\s{\mathbf{s}}
\def\S{\mathcal{S}}

\def\T{\mathcal{T}}

\def\Ubf{\mathbf{U}}
\def\u{\mathbf{u}}

\def\V{\mathcal{V}}

\def\var{\text{var}}

\def\y{\mathbf{y}}
\def\x{\mathbf{x}}

\def\X{\mathbf{X}}

\def\Y{\mathbf{Y}}
\def\z{\mathbf{z}}

\def\one{{\bf 1}}

\DeclareMathOperator{\trace}{trace}

\newcommand{\twonorm}[1]{\left\| #1 \right\|}

\usepackage{times}
\usepackage{graphicx} 
\usepackage{subfigure}

\usepackage{algorithm}
\usepackage{algorithmic}

\usepackage{hyperref}

\begin{document}

%
\title{Coding Method for Parallel Iterative Linear Solver}
%
%
%

\author{Yaoqing Yang,~\IEEEmembership{Student Member,~IEEE,}
        Pulkit~Grover,~\IEEEmembership{Senior Member,~IEEE,}
        and~Soummya~Kar
\thanks{Y. Yang, P. Grover and S. Kar are with the Department of Electrical and Computer Engineering, Carnegie Mellon University, Pittsburgh,
PA, 15213.} 
}

%
%

\markboth{}%
{Shell \MakeLowercase{\textit{et al.}}: Bare Demo of IEEEtran.cls for Journals}
%



\maketitle

\begin{abstract}
Computationally intensive distributed and parallel computing is often bottlenecked by a small set of slow workers known as stragglers. In this paper, we utilize the  emerging idea of ``coded computation'' to design a novel error-correcting-code inspired technique for solving linear inverse problems under specific iterative methods in a parallelized implementation affected by stragglers. Example applications include inverse problems in machine learning on graphs, such as personalized PageRank and sampling on graphs. We provably show that our coded-computation technique can reduce the mean-squared error under a computational deadline constraint. In fact, the ratio of mean-squared error of replication-based and coded techniques diverges to infinity as the deadline increases. Our experiments for personalized PageRank performed on real systems and real social networks show that this ratio can be as large as $10^4$. Further, unlike coded-computation techniques proposed thus far, our strategy combines outputs of all workers, including the stragglers, to produce more accurate estimates at the computational deadline. This also ensures that the accuracy degrades ``gracefully'' in the event that the number of stragglers is large. 
\end{abstract}

\begin{IEEEkeywords}
Distributed computing, error-correcting codes, straggler effect, iterative linear inverse, personalized PageRank
\end{IEEEkeywords}

\section{Introduction}\label{sec:Intro}

Although modern distributed computing systems have developed techniques for maintaining fault tolerance, the performance of such computing systems is often bottlenecked by a small number of slow workers. This ``straggling'' effect of the slow workers \cite{dean2013tail,joshi2014delay,wang2014efficient,wang2015using} is often addressed by replicating tasks across workers and using this redundancy to ignore some of the stragglers. Recently, concepts from error-correcting codes have been used for speeding up distributed computing  \cite{huang2012codes,lee2016speeding,tandon2016gradient,dutta2016short,ferdinand2016anytime,attia2016worst,li2016unified,reisizadehmobarakeh2017coded,li2017coding,yu2017polynomial,lee2017mds,lee2017high,lee2017coded,azian2017consensus}, which build on classical works on algorithm-based fault-tolerance \cite{huang1984algorithm}. The key idea is to treat slow workers as ``erasures'' and use error-correcting codes to retrieve the result after a subset of fast workers have finished. In some cases, (e.g.~\cite{lee2016speeding,dutta2016short} for matrix multiplications), coding techniques (\textit{i.e.}, techniques that utilize error-correcting codes) achieve scaling-sense speedups in average computation time compared to replication. In this work, we propose a novel coding-inspired technique to deal with the straggler effect in distributed computing of linear inverse problems using iterative solvers \cite{saad2003iterative}, such as personalized PageRank \cite{page1999pagerank,haveliwala2002topic} and signal recovery on large graphs \cite{wang2015local,narang2013localized,chen2015discrete}. We focus on iterative methods for solving these linear systems. For the personalized PageRank problem, we study the power-iteration method which is the most classical PageRank algorithm \cite{page1999pagerank}. For the problem of signal recovery on graphs, we study linear iterative algorithms such as the projected gradient descent algorithm \cite{wang2015local,narang2013localized,chen2015discrete}.

Existing algorithms that use coding to deal with stragglers treat straggling workers as ¡°erasures¡±, that is, they ignore computation results of the stragglers. Interestingly, in the iterative solvers for linear inverse problems, even if the computation result at a straggler has not converged, the proposed algorithm does not ignore the result, but instead combines it (with appropriate weights) with results from other workers. This is in part because the result of the iterative method converges gradually to the true solution. We use a small example shown in Fig.~\ref{fig:old_technique} to illustrate this idea. Suppose we want to solve two linear inverse problems with solutions $\x_1^*$ and $\x_2^*$. We ``encode the computation'' by adding an extra linear inverse problem, the solution of which is $\x_1^*+\x_2^*$ (see Section~\ref{sec:alg}), and distribute these three problems to three workers. Using this method, the solutions $\x_1^*$ and $\x_2^*$ can be obtained from the results of any combination of two fast workers that are first to come close to their solutions.

But what if we have a computational deadline, $T_{dl}$, by which only one worker converges to its solution? In that case, the natural extension of existing coded-computation strategies (e.g., \cite{lee2016speeding}) will declare a computation failure because it needs at least two workers to respond. However, our strategy  does not require convergence: even intermediate results from workers, as long as they are received, can be utilized to estimate solutions. In other words, our strategy degrades gracefully as the number of stragglers increases, or as the deadline is pulled earlier. Indeed, we will show that it is suboptimal to ignore stragglers as erasures, and will design strategies that treat the difference from the optimal solution as ``soft'' additive noise (see Fig.~\ref{fig:lee}, and Section \ref{sec:alg}). We use an algorithm that is similar to weighted least squares algorithm for decoding, giving each worker a weight based on its proximity to convergence.
In this way, we can expect to fully utilize the computation results from all workers and obtain better speedup. 

\begin{figure}
  \centering
  \includegraphics[scale=0.2]{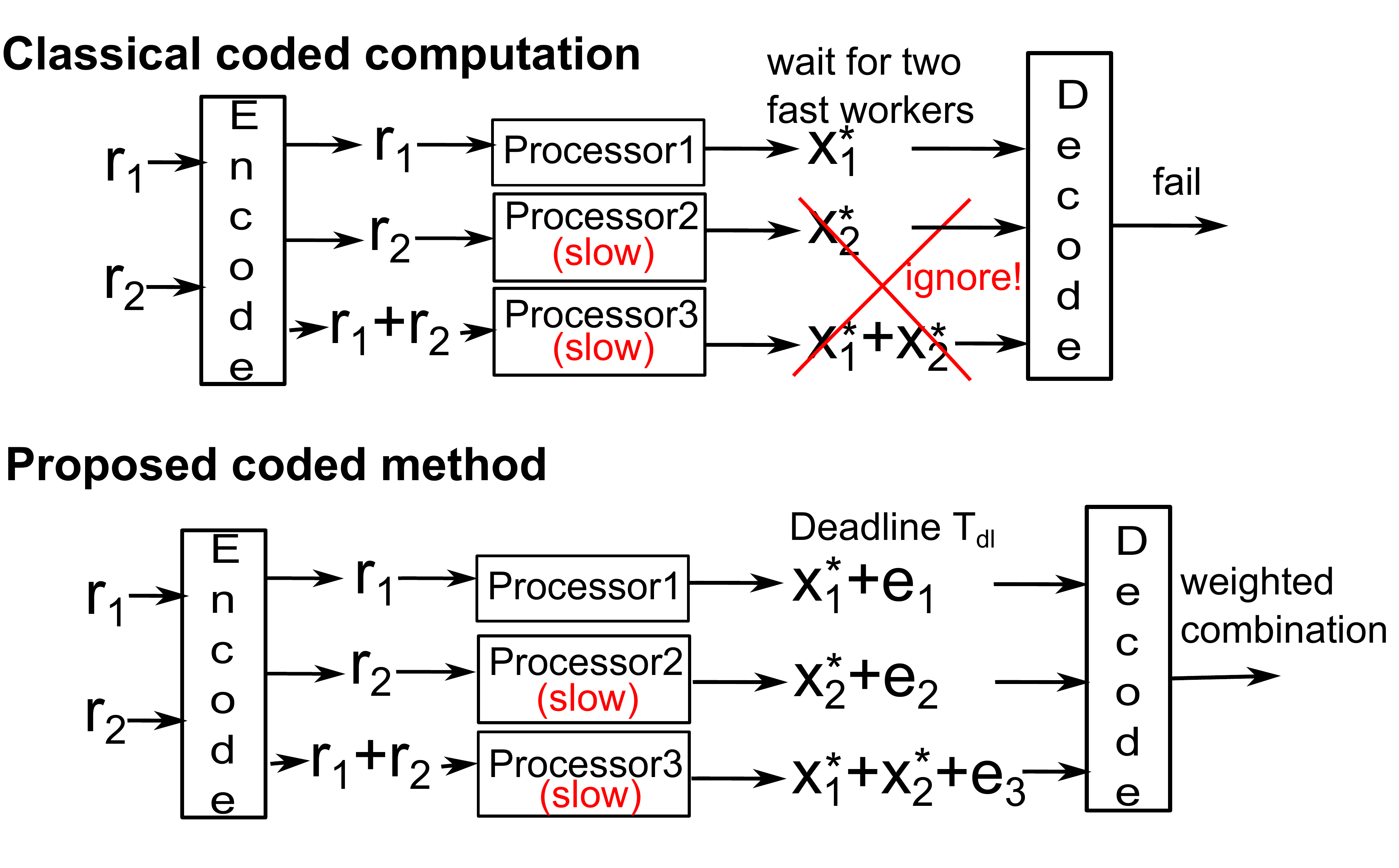}\\
  \caption{A toy example of the comparison between the existing scheme in \cite{lee2016speeding} and the proposed algorithm.\vspace{-5mm}} \label{fig:old_technique}
\end{figure}
Theoretically, we show that for a specified deadline time $T_\text{dl}$, under certain conditions on worker speed, the coded linear inverse solver using structured codes has smaller mean squared error than the replication-based linear solver (Theorem~\ref{thm:beat_replication}). In fact, when the computation time $T_\text{dl}$ increases, the ratio of the mean-squared error (MSE) of replication-based and coded linear solvers can get arbitrarily large (Theorem~\ref{thm:TtoInf})!

For validation of our theory, we performed experiments to compare the performance of coded and replication-based personalized PageRank (respectively using coded and replication-based power-iteration method) on the Twitter and Google Plus social networks under a deadline on computation time using a given number of workers on a real computation cluster (Section~\ref{sec:experimenl}). We observe that the MSE of coded PageRank is smaller than that of replication by a factor of $10^4$ at $T_{dl}= 2$ seconds. From an intuitive perspective, the advantage of coding over replication is that coding utilizes the diversity of all heterogeneous workers. To compare with existing coded technique in \cite{lee2016speeding}, we adapt its technique to inverse problems by inverting only the partial results from the fast workers. However, from our experiments, if only the results from the fast workers are used, the error amplifies due to inverting an ill-conditioned submatrix during decoding (Section~\ref{sec:experimenl}). This ill-conditioning issue of real-number erasure codes has also been recognized in a recent communication problem \cite{haikin2016analog}. In contrast, our novel way of combining all the partial results including those from the stragglers helps bypass the difficulty of inverting an ill-conditioned matrix.

The focus of this work is on utilizing computations to deliver the minimal MSE in solving linear inverse problems. Our algorithm does not reduce the communication cost. However, because each worker performs sophisticated iterative computations, such as the power-iteration computations in our problem, the time required for computation dominates that of communication (see Section~\ref{sec:com_vs_comp}).
This is unlike some recent works (e.g.\cite{dimakis2010network,sathiamoorthy2013xoring,maddah2015decentralized,li2015coded,li2016unified,tandon2016gradient}) where communication costs are observed to dominate. However, in these works, the computational tasks at each worker are simpler, and sometimes the communication cost is large because of data-shuffling (e.g. in MapReduce \cite{dean2008mapreduce}).

Finally, we summarize our main contributions in this paper:
\begin{itemize}
\item We propose a coded algorithm for distributed computation for multiple instances of a linear inverse problem;
\item We theoretically analyze the mean-squared error of coded, uncoded and replication-based iterative linear solvers under a deadline constraint;
\item We compare mean-squared error of coded linear solver with uncoded and replication-based linear solver and show scaling sense advantage in theory and orders of magnitude smaller error in data experiments;
\item This is the first work that treats stragglers as soft errors instead of erasures, which leads to graceful degradation in the event that the number of stragglers is large.
\end{itemize}

\section{System Model and Problem Formulation}\label{sec:system_model}
In this section, we present the model of parallel linear systems that we will apply the idea of error correcting codes. Then, we provide two applications that can be directly formulated in the form of parallel linear systems.

\subsection{Preliminaries on Solving Parallel Linear Systems using Iterative Methods}\label{sec:par_sys}
Consider the problem of solving $k$ inverse problems with the same linear transform matrix and different inputs $\rbf_i$:
\begin{equation}\label{eqn:general_inverse}
\Mbf\x_i=\rbf_i,i=1,2,\ldots k.
\end{equation}
When $\Mbf$ is a square matrix, the closed-form solution is
\begin{equation}\label{eqn:general_inverse_solution}
\x_i=\Mbf^{-1}\rbf_i.
\end{equation}
When $\Mbf$ is a non-square matrix, the solution to \eqref{eqn:general_inverse} is interpreted as
\begin{equation}\label{eqn:general_l2}
\x_i=\arg\min\twonorm{\Mbf\x-\rbf_i}^2+\lambda\twonorm{\x}^2, i=1,2,\ldots k,
\end{equation}
with an appropriate regularization parameter $\lambda$. The closed-form solution of \eqref{eqn:general_l2} is
\begin{equation}\label{eqn:general_l2_solution}
\x_i=(\Mbf^\top\Mbf+\lambda \I)^{-1}\Mbf^\top\rbf_i.
\end{equation}
Computing matrix inverse in \eqref{eqn:general_inverse_solution} or \eqref{eqn:general_l2_solution} directly is often hard and commonly used methods are often iterative. In this paper, we study two ordinary iterative methods, namely the Jacobian method for solving \eqref{eqn:general_inverse} and the gradient descent method for solving \eqref{eqn:general_l2}.
\subsubsection{The Jacobian Method for Square System}\label{sec:Jacobian}
For a square matrix $\Mbf$, one can decompose $\Mbf=\D+\Lbf$, where $\D$ is diagonal. Then, the Jacobian iteration is written as
\begin{equation}\label{eqn:Jacobian}
\x_i^{(l+1)}=\D^{-1}(\rbf_i-\Lbf \x_i^{(l)}).
\end{equation}
Under certain conditions of $\D$ and $\Lbf$ (see \cite[p.115]{saad2003iterative}), the computation result converges to the true solution.
\subsubsection{Gradient Descent for Non-square System}\label{sec:gradient_descent}
For the $\ell$-2 minimization problem \eqref{eqn:general_l2}, the gradient descent solution has the form
\begin{equation}
\x_i^{(l+1)}=((1-\lambda)\I- \epsilon\Mbf^\top\Mbf)\x_i^{(l)}+\epsilon\Mbf^\top\rbf_i,
\end{equation}
where $\epsilon$ is an appropriate step-size.

From these two problem formulations, we can see that the iterative methods have the same form
\begin{equation}\label{eqn:power_iter}
\x_i^{(l+1)}=\B\x_i^{(l)}+\Kbf\rbf_i,i=1,2,\ldots k,
\end{equation}
for two appropriate matrices $\B$ and $\Kbf$. Denote by $\x_i^*$ the solution to \eqref{eqn:general_inverse} or \eqref{eqn:general_l2}. Then,
\begin{equation}\label{eqn:fixed_poinl}
\x_i^*=\B\x_i^*+\Kbf\rbf_i,i=1,2,\ldots k.
\end{equation}
Then, from \eqref{eqn:fixed_poinl} and \eqref{eqn:power_iter}, the computation error $\e_i^{(l)}=\x_i^{(l)}-\x_i^*$ satisfies
\begin{equation}\label{eqn:error_iter}
\e_i^{(l+1)}=\B\e_i^{(l)}.
\end{equation}

\subsection{Distributed Computing and the Straggler Effect}
Consider solving $k$ linear inverse problems in $k$ parallel workers using the iterative method \eqref{eqn:power_iter}, such that each processor solves one problem. Due to the straggler effect, the computation of the linear inverse problem on different workers can have different computation speeds. Suppose after the deadline time $T_\text{dl}$, the $i$-th worker has completed $l_i$ iterations in \eqref{eqn:power_iter}. Then, the residual error at the $i$-th worker is
\begin{equation}\label{eqn:error_iter_A}
\e_i^{(l_i)}=\B^{l_i}\e_i^{(0)}.
\end{equation}

For our theoretical results, we sometimes need the following assumption.

\begin{assumption}\label{ass:general}
We assume that the optimal solutions $\x_i^*,i=1,2,\ldots k$ are i.i.d.
\end{assumption}
Denote by $\boldsymbol{\mu}_E$ and $\C_E$ respectively the mean and the covariance of each $\x_i^*$. Assume we start with the initial estimate $\x_i^{(0)}=\boldsymbol{\mu}_E$, which can be estimated from data. Then, $\e_i^{(0)}=\x_i^{(0)}-\x_i^*$ has mean $\0_N$ and covariance $\C_E$.

Note that Assumption~\ref{ass:general} is equivalent to the assumption that the inputs $\rbf_i,i=1,2,\ldots k$ are i.i.d., because the input and the true solution for a linear inverse problem are related by a linear transform. We provide an extension of this i.i.d. assumption in Section~\ref{sec:not_iid}.

\subsection{Two Motivating Examples}\label{sec:PG_and_straggler}

Now we study two examples of the general linear inverse problems. Both of them related to data processing of graphs.

\subsubsection{PageRank as a Square System}
For a directed graph $\G=(\V,\E)$ with node set $\V$ and edge set $\E$, the PageRank algorithm aims to measure the importance of the nodes in $\V$ by computing the stationary distribution of a discrete-time ``random walk with restart'' on the graph that mimics the behavior of a web surfer on the Internet. At each step, with probability $1-d$, the random walk chooses a random neighbor on the graph with uniform probability to proceed to (e.g. $d= 0.15$ in~\cite{page1999pagerank}). With probability $d$, it jumps to an arbitrary node in the graph. The probability $d$ is often called the ``teleport probability''. From \cite{page1999pagerank}, the problem of computing the stationary distribution is equivalent to solving the following linear problem
\begin{equation}\label{eqn:pg}
\x=\frac{d}{N}\one_N+(1-d)\A\x,
\end{equation}
where $N$ is the number of nodes and $\mathbf{A}$ is the column-normalized adjacency matrix, i.e., for a directed edge $v_i\to v_j$, $\A_{ij}=\frac{1}{\text{deg}(v_j)}$, where $\text{deg}(v_j)$ is the in-degree of $v_j$.

The personalized PageRank problem \cite{haveliwala2002topic} considers a more general linear equation
\begin{equation}\label{eqn:personalized_pg_0}
\x=d \rbf+(1-d)\mathbf{A}\x,
\end{equation}
for any possible vector $\rbf\in \R^N$ that satisfies $\one^\top\rbf=1$. Compared to the original PageRank problem \cite{page1999pagerank},  personalized PageRank \cite{haveliwala2002topic} utilizes both the structure of the graph \textit{and} the personal preferences of different users. The solution $\x$ is also the limiting distribution of a random walk, but with different restarting distribution. That is, with probability $d$, instead of jumping to each node with uniform probability, the random walk jumps to different nodes according to distribution $\rbf$. Intuitively, difference in the vector $\rbf$ represents different preferences of web surfers.

A classical method to solve PageRank is power-iteration, which iterates the following computation until convergence (usually with initial condition $\x_i^{(0)}=\frac{1}{N}\one_N$):
\begin{equation}\label{eqn:power_iter_general}
\x^{(l+1)}=d \rbf+(1-d)\mathbf{A}\x^{(l)}.
\end{equation}
One can see that the power-iteration method is exactly the same as the Jacobian iteration \eqref{eqn:Jacobian}.

\subsubsection{``Graph Signal Reconstruction'' as a Non-Square System}

The emerging field of signal processing on graphs \cite{shuman2013emerging,sandryhaila2013discrete} is based on a neat idea of treating ``values'' associated with the nodes of a graph as ``a signal supported on a graph'' and apply techniques from signal processing to solve problems such as prediction and detection. For example, the number of cars at different road intersections on the Manhattan road network can be viewed as a graph signal on the road graph $\G=(\V,\E)$ where $\V$ is the set of road intersections and $\E$ is the set of road segments. For a directed or undirected graph $\G=(\V,\E)$, the graph signal has the same dimension as the number of nodes $|\V|$ in the graph, i.e., there is only one value associated with each node.

One important problem of graph signal processing is that of recovering the values on the remaining nodes of a graph given the values sampled on a particular node subset $\S\subset \V$ under the assumption of ``bandlimited'' graph signal \cite{wang2015local,narang2013localized,chen2015discrete}. One application of graph signal reconstruction can be that of reconstructing the entire traffic flow using the observations from a few cameras at some road intersections \cite{ChenYFK_GlobalSIP_16}. The graph signal reconstruction problem can be formulated as a least-square solution of the following linear system (see equation (5) in \cite{narang2013localized})
\begin{equation}
\left[\begin{matrix}
\f(\S)\\
\f(\S^c)
\end{matrix}\right]=
\left[\begin{matrix}
\u_1(\S) &\u_2(\S)&\ldots &\u_w(\S)\\
\u_1(\S^c)&\u_2(\S^c)&\ldots &\u_w(\S^c)
\end{matrix}\right]
\left[\begin{matrix}
\alpha_1\\
\alpha_2\\
\vdots\\
\alpha_w.
\end{matrix}\right],
\end{equation}
where $\f(\S)$ is the given part of the graph signal $\f$ on the set $\S$, $\f(\S^c)$ is the unknown part of the graph signal $\f$ to be reconstructed, $[\alpha_1,\alpha_2,\ldots \alpha_w]$ are unknown coefficients of the graph signal $\f=\left[\begin{matrix}
\f(\S)\\
\f(\S^c)
\end{matrix}\right]$ represented in the subspace spanned by $\u_1,\u_2,\ldots \u_w$ which are the first $w$ eigenvectors of the graph Laplacian matrix. It is shown in \cite{narang2013localized} that this least-square reconstruction problem can be solved using an iterative linear projection method (see equation (11) in \cite{narang2013localized})
\begin{equation}
\f_{k+1}=\Ubf\Ubf^\top \left(\f_k+\J^\top \J(\left[\begin{matrix}
\f(\S)\\
\0
\end{matrix}\right]-\f_k)\right),
\end{equation}
where $\Ubf=[\u_1,\u_2,\ldots \u_w]$. This iteration can be shown to converge to the least square solution given by
\begin{equation}
\f(\S^c)=(\Ubf)_{\S^c}((\Ubf)_{\S^c}^\top (\Ubf)_{\S^c})^{-1}(\Ubf)_{\S^c}^t\f(\S).
\end{equation}
Note that we may want to reconstruct multiple instances of graph signal, such as road traffic at different time, which brings in the formulation of distributed computing.

\subsection{Preliminaries on Error Correcting Codes}\label{sec:ecc}

We borrow concepts and terminology from error-correcting codes \cite{peterson1972error}. E.g., we will use ``encode'' and ``decode'' to denote preprocessing and post-processing before and after parallel computation. Although classically encoding and (specially) decoding can be nonlinear operations, in this paper, the encoder multiplies the inputs to the parallel workers with a ``generator matrix'' $\Gbf$ and the decoder multiplies the outputs of the workers with a ``decoding matrix'' $\Lbf$. We call a code an ($n,k$) code if the generator matrix has size $k\times n$.

In this paper, we often use generator matrices $\Gbf$ with orthonormal rows, which means
\begin{equation}
\Gbf_{k\times n}\Gbf_{n\times k}^\top=\I_k.
\end{equation}
An example of such a matrix is the submatrix formed by any $k$ rows of an $n\times n$ orthonormal matrix (e.g., a Fourier matrix). Under this assumption, $\Gbf_{k\times n}$ can be augmented to form an $n\times n$ orthonormal matrix using another matrix $\H_{(n-k)\times n}$, i.e. the square matrix $\Fbf_{n\times n}=\left[\begin{matrix}
\Gbf_{k\times n}\\
\H_{(n-k)\times n}
\end{matrix}\right]$ satisfies $\Fbf^\top \Fbf=I_n$.
This structure assumption is not necessary when we compare the error exponents of different linear inverse algorithms when the computation time $T_\text{dl}$ goes to infinity (e.g., coded and uncoded). However, it is useful when we present theorems on finite $T_\text{dl}$.

\section{Coded Distributed Computing of Linear Inverse Problems}
\subsection{The Coded Linear Inverse Algorithm}\label{sec:alg}

\begin{algorithm}[!h]
   \caption{Coded Distributed Linear Inverse}\label{alg:pg}
\begin{algorithmic}
   \STATE {\bfseries Input:} Input vectors $[\rbf_1,\rbf_2,\ldots,\rbf_k]$, generator matrix $\Gbf_{k\times n}$, the linear system matrices $\B$ and $\Kbf$ defined in \eqref{eqn:power_iter}.
   \STATE {\bfseries Initialize (Encoding):} Encode the input vectors $[\rbf_1,\rbf_2,\ldots,\rbf_k]$ and the initial estimates by multiplying with the generator matrix $\Gbf$:
\begin{equation}\label{eqn:code_r}
[\s_1,\s_2,\ldots,\s_n]=[\rbf_1,\rbf_2,\ldots,\rbf_k]\cdot\Gbf.
\end{equation}
\begin{equation}\label{eqn:code_x0}
[\y_1^{(0)},\y_2^{(0)},\ldots,\y_n^{(0)}]=[\x_1^{(0)},\x_2^{(0)},\ldots,\x_k^{(0)}]\cdot\Gbf.
\end{equation}
	\STATE {\bfseries Parallel Computing:}
   \FOR{$i=1$ {\bfseries to} $n$}
   \STATE Send $\s_i$ and $\y_i^{(0)}$ to the $i$-th worker. Compute the solution of \eqref{eqn:general_inverse} or \eqref{eqn:general_l2} using the specified iterative method \eqref{eqn:power_iter} with initial estimate $\y_i^{(0)}$ at each worker in parallel until a deadline $T_\text{dl}$.
   \ENDFOR
   \STATE After $T_\text{dl}$, collect all linear inverse results $\y_i^{(l_i)}$ from these $n$ workers. The superscript $l_i$ in $\y_i^{(l_i)}$ represents that the $i$-th worker finished $l_i$ iterations within time $T_\text{dl}$. Denote by $\Y^{(T_\text{dl})}$ the collection of all results
   \begin{equation}\label{eqn:pg_results}
\Y^{(T_\text{dl})}_{N\times n}=[\y_1^{(l_1)},\y_2^{(l_2)},\ldots,\y_n^{(l_n)}].
\end{equation}

\STATE {\bfseries Post Processing (Decoding):}
\STATE Compute an estimate of the linear inverse solutions using the following matrix multiplication:
\begin{equation}\label{eqn:weighted_least_square}
\hat\X^\top=\Lbf\cdot(\Y^{(T_\text{dl})})^\top:=(\Gbf \boldsymbol{\Lambda}^{-1} \Gbf^\top)^{-1}\Gbf\boldsymbol{\Lambda}^{-1}(\Y^{(T_\text{dl})})^\top,
\end{equation}
\textcolor{black}{where the estimate} $\hat\X_{N\times k}=[\hat\x_1,\hat\x_2,\ldots,\hat\x_k]$, the matrix $\boldsymbol{\Lambda}$ is
\begin{equation}\label{eqn:Gamma}
\boldsymbol{\Lambda}=\diag\left[\trace(\C(l_1)),\ldots,\trace(\C(l_n))\right],
\end{equation}
where the matrices $\C(l_i),i=1,\ldots,n$ are defined as
\begin{equation}\label{eqn:CTI}
\C(l_i)=\B^{l_i}\C_E(\B^\top)^{l_i}.
\end{equation}
In computation of $\boldsymbol{\Lambda}$, if $\trace(\C(l_i))$ are not available, one can use estimates of this trace as discussed in Section~\ref{sec:Lambda}.
\end{algorithmic}
\end{algorithm}

The proposed coded linear inverse algorithm is shown in Algorithm \ref{alg:pg}. The algorithm has three stages: preprocessing (encoding) at the central controller, parallel computing at $n$ parallel workers, and post-processing (decoding) at the central controller. As we show later in the analysis of computing error, the entries in the diagonal matrix $\boldsymbol{\Lambda}$ are the expected mean-squared error at each worker prior to decoding. The decoding matrix $\Lbf_{k\times n}$ in the decoding step \eqref{eqn:weighted_least_square} is chosen to be $(\Gbf \boldsymbol{\Lambda}^{-1} \Gbf^\top)^{-1}\Gbf\boldsymbol{\Lambda}^{-1}$ to reduce the mean-squared error of the estimates of linear inverse solutions by assigning different weights to different workers based on the estimated accuracy of their computation (which is what $\boldsymbol{\Lambda}$ provides).

\subsection{\textcolor{black}{Bounds on Performance of the Coded Linear Inverse Algorithm}}\label{sec:finite_time_analysis}
Define $\mathbf{l}=[l_1,l_2,\ldots l_n]$ as the vector of the number of iterations at all workers. We use the notation $\Ep[\cdot|\mathbf{l}]$ to denote the conditional expectation taken with respect to the randomness of the optimal solution $\x_i^*$ (see Assumption \ref{ass:general}) but conditioned on fixed iteration number $l_i$ at each worker, i.e., for a random variable $X$,
\begin{equation}
\Ep[X|\mathbf{l}]=\Ep[X|l_1,l_2,\ldots l_n].
\end{equation}

\begin{theorem}\label{thm:upper_bound}
Define $\Ebf=\hat\X-\X^*$, i.e., the error of the decoding result \eqref{eqn:weighted_least_square}. Assuming that the solutions for each linear inverse problem are chosen i.i.d. (across all problems) according to a distribution with covariance $\C_E$. Then, the error covariance of $\Ebf$ satisfies
\begin{equation}\label{eqn:error_ub}
\Ep[\twonorm\Ebf^2|\mathbf{l}]\le\sigma_\text{max}(\Gbf^\top\Gbf)\trace\left[(\Gbf \boldsymbol{\Lambda}^{-1} \Gbf^\top)^{-1}\right],
\end{equation}
where the norm $\twonorm\cdot$ is the Frobenius norm, $\sigma_\text{max}(\Gbf^\top\Gbf)$ is the maximum eigenvalue of $\Gbf^\top\Gbf$ and the matrix $\boldsymbol{\Lambda}$ is defined in \eqref{eqn:Gamma}. Further, when $\Gbf$ has orthonormal rows,
\begin{equation}\label{eqn:error_ub_Gorth}
\Ep[\twonorm\Ebf^2|\mathbf{l}]\le\trace\left[(\Gbf \boldsymbol{\Lambda}^{-1} \Gbf^\top)^{-1}\right],
\end{equation}
\end{theorem}
\begin{proof}
See appendix Section~\ref{app:upper_bound} for a detailed proof. Here we provide the main intuition by analyzing a ``scalar version'' of the linear inverse problem, in which case the matrix $\B$ is equal to a scalar $a$. In appendix Section~\ref{app:upper_bound}, we make the generalization to all $\B$ matrices using matrix vectorization.

For $\B=a$, the inputs and the initial estimates in \eqref{eqn:code_r} and \eqref{eqn:code_x0} are vectors instead of matrices. As we show in appendix Section~\ref{app:upper_bound}, if we encode both the inputs and the initial estimates using \eqref{eqn:code_r} and \eqref{eqn:code_x0}, we also ``encode'' the error
\begin{equation}\label{eqn:up_der1}
[\epsilon_1^{(0)},\epsilon_2^{(0)},\ldots,\epsilon_n^{(0)}]=[e_1^{(0)},e_2^{(0)},\ldots,e_k^{(0)}]\cdot\Gbf=:\Ebf_0\Gbf,
\end{equation}
where $\epsilon_i^{(0)}=y_i^{(0)}-y_i^*$ is the initial error at the $i$-th worker, $e_i^{(0)}=x_i^{(0)}-x_i^*$ is the initial error of the $i$-th linear inverse problem, and $\Ebf_0:=[e_1^{(0)},e_2^{(0)},\ldots e_k^{(0)}] $. Suppose $\text{var}[e_i^{(0)}]=c_e$, which is a scalar version of $\C_E$ after Assumption~\ref{ass:general}. From \eqref{eqn:error_iter_A}, the error satisfies:
\begin{equation}\label{eqn:up_der2}
\epsilon_i^{(l_i)}=a^{l_i}\epsilon_i^{(0)},i=1,2,\ldots n.
\end{equation}
Denote by $\D=\diag\{a^{l_1},a^{l_2},\ldots a^{l_n}\}$. Therefore, from \eqref{eqn:up_der1} and \eqref{eqn:up_der2}, the error before the decoding step \eqref{eqn:weighted_least_square} can be written as
\begin{equation}
\begin{split}
[\epsilon_1^{(l_1)},\epsilon_2^{(l_2)},\ldots \epsilon_n^{(l_n)}]
=&[\epsilon_1^{(0)},\epsilon_2^{(0)},\ldots \epsilon_n^{(0)}]\cdot \D=\Ebf_0\Gbf\D.
\end{split}
\end{equation}
We can show (see appendix Section~\ref{app:upper_bound} for details) that after the decoding step \eqref{eqn:weighted_least_square}, the error vector is also multiplied by the decoding matrix $\Lbf=(\Gbf \boldsymbol{\Lambda}^{-1} \Gbf^\top)^{-1}\Gbf\boldsymbol{\Lambda}^{-1}$:
\begin{equation}
\begin{split}
\Ebf^\top&=\Lbf\left[\epsilon_1^{(l_1)},\epsilon_2^{(l_2)},\ldots \epsilon_n^{(l_n)}\right]^\top=\Lbf\D^\top \Gbf^\top\Ebf_0^\top.
\end{split}
\end{equation}
Thus,
\begin{equation}
\begin{split}
\Ep[\twonorm\Ebf^2|\mathbf{l}]
=&\Ep[\trace[\Ebf^\top\Ebf]|\textbf{l}]\\
=&\trace[\Lbf\D^\top \Gbf^\top\Ep[\Ebf_0^\top\Ebf_0|\textbf{l}] \Gbf\D\Lbf^\top]\\
\overset{(a)}{=}&\trace[\Lbf\D^\top \Gbf^\top c_e\I_k   \Gbf\D\Lbf^\top]\\
=&c_e\trace[\Lbf\D^\top \Gbf^\top\Gbf\D\Lbf^\top]\\
\overset{(b)}{\le}&c_e\sigma_\text{max}(\Gbf^\top\Gbf)\trace[\Lbf\D^\top \D\Lbf^\top]\\
=&\sigma_\text{max}(\Gbf^\top\Gbf)\trace[\Lbf (c_e\D^\top \D)\Lbf^\top]\\
\overset{(c)}{=} & \sigma_\text{max}(\Gbf^\top\Gbf)\trace[\Lbf\boldsymbol{\Lambda}\Lbf^\top]\\
\overset{(d)}{=} &\sigma_\text{max}(\Gbf^\top\Gbf)\trace[(\Gbf \boldsymbol{\Lambda}^{-1} \Gbf^\top)^{-1}],
\end{split}
\end{equation}
where (a) holds because $\Ebf_0:=[e_1^{(0)},e_2^{(0)},\ldots e_k^{(0)}] $ and $\text{var}[e_i^{(0)}]=c_e$, (b) holds because $\Gbf^\top\Gbf\preceq \sigma_\text{max}(\Gbf^\top\Gbf)\I_n$, (c) holds because $c_e\D^\top \D=\boldsymbol{\Lambda}$, which is from the fact that for a scalar linear system matrix $\B=a$, the entries in the $\boldsymbol{\Lambda}$ matrix in \eqref{eqn:Gamma} satisfy
\begin{equation}
\trace(\C(l_i))=a^{l_i}c_e(a^\top)^{l_i}=c_ea^{2l_i},
\end{equation}
which is the same as the entries in the diagonal matrix $c_e\D^\top\D$. Finally, (d) is obtained by directly plugging in $\H:=(\Gbf \boldsymbol{\Lambda}^{-1} \Gbf^\top)^{-1}\Gbf\boldsymbol{\Lambda}^{-1}$. Finally, inequality \ref{eqn:error_ub_Gorth} holds because when $\Gbf$ has orthonormal rows, $\sigma(\Gbf^\top\Gbf)=\sigma(\Gbf^\top\Gbf)=1$.
\end{proof}

\begin{corollary}\label{cor:upper_bound}
Suppose the i.i.d. Assumption \ref{ass:general} holds and the matrix $\Gbf_{k\times n}$ is a submatrix of an $n\times n$ orthonormal matrix, i.e. there exists a matrix $\Fbf_{n\times n}=\left[\begin{matrix}
   \Gbf_{k\times n}\\
   \H_{(n-k)\times n}
\end{matrix}\right]$ satisfies $\Fbf^\top \Fbf=\I_n$. Assume that the symmetric matrix $\Fbf \boldsymbol{\Lambda}\Fbf^\top$ has the block form
\begin{equation}\label{eqn:block_form}
\Fbf \boldsymbol{\Lambda}\Fbf^\top=   \left[
\begin{matrix}
   {\J_1} & {\J_2} \\
   {\J_2^\top} & {\J_4} \\
\end{matrix}  \\
\right]_{n\times n},
\end{equation}
that is, $(\J_1)_{k\times k}$ is $\Gbf \boldsymbol{\Lambda} \Gbf^\top$, $(\J_2)_{k\times (n-k)}$ is $\Gbf \boldsymbol{\Lambda} \H^\top$, and $(\J_4)_{(n-k)\times (n-k)}$ is $\H \boldsymbol{\Lambda} \H^\top$. Then, we have
\begin{equation}\label{eqn:error_ub2}
\Ep[\twonorm\Ebf^2|\mathbf{l}]\le\trace(\J_1)-\trace(\J_2\J_4^{-1}\J_2^\top).
\end{equation}
\end{corollary}
\begin{proof}
The proof essentially relies on the Schur complement. See appendix Section~\ref{app:Schur} for details.
\end{proof}

\subsection{Bounds on the Mean-squared Error beyond the i.i.d. Case}\label{sec:not_iid}

\textcolor{black}{Until now, we based our analysis on the i.i.d. assumption~\ref{ass:general}. For the PageRank problem discussed in Section~\ref{sec:PG_and_straggler}, this assumption means that the PageRank queries are independent across different users. Although the case when the PageRank queries are arbitrarily correlated is hard to analyze, we may still provide concrete analysis for some specific cases. For example, a reasonable case when the PageRank queries are correlated with each other is when these queries are all affected by some ``common fashion topic'' that the users wish to search for. In mathematics, we can model this phenomenon by assuming that the solutions to the $i$-th linear inverse problem satisfies}
\begin{equation}
\x_i^*=\bar{\x}+\z_i,
\end{equation}
for some random vector $\bar{\x}$ and an i.i.d. vector $\z_i$ across different queries (different $i$). The common part $\bar{\x}$ is random because the common fashion topic itself can be random. This model can be generalized to the following ``stationary'' model.

\begin{assumption}\label{ass:stationary}
Assume the solutions $\x_i^*$'s of the linear inverse problems have the same mean $\boldsymbol{\mu}_E$ and stationary covariances, i.e.,
\begin{equation}\label{eqn:stationary1}
\Ep[\x_i^*(\x_i^*)^\top]=\C_E+\C_\text{Cor},\forall 1\le i\le k,
\end{equation}
\begin{equation}\label{eqn:stationary2}
\Ep[\x_i^*(\x_j^*)^\top]=\C_\text{Cor},\forall 1\le i,j\le k.
\end{equation}
\end{assumption}
Under this assumption, we have to change the coded linear inverse algorithm slightly. The details are shown in Algorithm~\ref{alg:pg_stationary}.

\begin{algorithm}[!h]
   \caption{Coded Distributed Linear Inverse (Stationary Inputs)}\label{alg:pg_stationary}

Call Algorithm~\ref{alg:pg} but replace the $\boldsymbol{\Lambda}$ matrix with
\begin{equation}\label{eqn:tildeLambda}
\tilde{\boldsymbol{\Lambda}}=\sigma_\text{max}(\Gbf^\top\Gbf)\boldsymbol{\Lambda}+\diag\{\Gbf^\top \one_k\}\cdot \boldsymbol{\Psi}\cdot \diag\{\Gbf^\top \one_k\}^\top,
\end{equation}
where $\sigma_\text{max}(\Gbf^\top\Gbf)$ is the maximum eigenvalue of $\Gbf^\top\Gbf$, and $\boldsymbol\Psi_{n\times n}=[\Psi_{i,j}]$ satisfies
\begin{equation}\label{eqn:tildeLambda_ij}
\Psi_{i,j}=\trace[\B^{l_i}\C_\text{cor}(\B^\top)^{l_j}].
\end{equation}
\end{algorithm}

For the stationary version, we can have the counterpart of Theorem~\ref{thm:upper_bound} as follows. Trivial generalizations include arbitrary linear scaling $\x_i^*=\alpha_i\bar{\x}+\beta_i\z_i$ for scaling constants $\alpha_i$ and $\beta_i$.

\begin{theorem}\label{thm:upper_bound_stationary}
Define $\Ebf=\hat\X-\X^*$, i.e., the error of the decoding result \eqref{eqn:weighted_least_square} by replacing $\boldsymbol\Lambda$ defined in \eqref{eqn:Gamma} with $\tilde{\boldsymbol{\Lambda}}$ in \eqref{eqn:tildeLambda}. Assuming that the solutions for all linear inverse problems satisfy Assumption \ref{ass:stationary}. Then, the error covariance of $\Ebf$ satisfies
\begin{equation}\label{eqn:error_ub_stationary}
\begin{split}
\Ep[\twonorm\Ebf^2|\mathbf{l}]\le&\trace\left[(\Gbf \tilde{\boldsymbol{\Lambda}}^{-1} \Gbf^\top)^{-1}\right].
\end{split}
\end{equation}
where the norm $\twonorm\cdot$ is the Frobenius norm.
\end{theorem}
\begin{proof}
See appendix Section~\ref{app:correlated}.
\end{proof}

In Section \ref{sec:theory}, we compare coded, uncoded and replication-based linear inverse schemes under the i.i.d. assumption. However, we include one experiment in Section~\ref{sec:experimenl} to show that Algorithm~\ref{alg:pg_stationary} also works in the stationary case.

\section{Comparison with Uncoded Schemes and Replication-based Schemes}\label{sec:theory}

Here, we often assume (we will state explicitly in the theorem) that the number of iterations $l_i$ at different workers are i.i.d.. $\Ep_f[\cdot]$ denotes expectation on randomness of both the linear inverse solutions $\x^*_i$ and the number of iterations $l_i$.
\begin{assumption}\label{ass:time_iid}
Within time $T_\text{dl}$, the number of iterations of linear inverse computations at each worker follows an i.i.d. distribution $l_i\sim f(l)$.
\end{assumption}

\subsection{Comparison between the coded and uncoded linear inverse before a deadline}
First, we compare the coded linear inverse scheme with an uncoded scheme, in which case we use the first $k$ workers to solve $k$ linear inverse problems in \eqref{eqn:fixed_poinl} without coding. The following theorem quantifies the overall mean-squared error of the uncoded scheme given $l_1,l_2,\ldots,l_k$. The proof is in appendix Section \ref{app:uncoded_replication}.
\begin{theorem}\label{thm:uncoded}
In the uncoded scheme, the overall error is
\begin{equation}\label{eqn:error_uncoded}
\begin{split}
\hspace{-0.1in}\Ep\left[\twonorm{\Ebf_{\text{uncoded}} }^2|\textbf{l}\right]&=\Ep\left[\left.\twonorm{[{\e_1^{(l_1)}},{\e_2^{(l_2)}}\ldots,{\e_k^{(l_k)}}]}^2\right|\mathbf{l}\right]\\
&=\sum_{i=1}^k \trace\left(\C(l_i)\right).
\end{split}
\end{equation}
Further, when the i.i.d. Assumption \ref{ass:time_iid} holds,
\begin{equation}\label{eqn:error_uncoded_f}
\begin{split}
&\Ep_f \left[\twonorm{\Ebf_{\text{uncoded}} }^2\right]= k\Ep_f[\trace(\C(l_1))].
\end{split}
\end{equation}
\end{theorem}

Then, we compare the overall mean-squared error of coded and uncoded linear inverse algorithms. \textbf{Note that this comparison is not fair} because the coded algorithm uses more workers than uncoded. However, we still include Theorem~\ref{thm:beat_uncoded} because we need it for the fair comparison between coded and replication-based linear inverse.

\begin{theorem}\label{thm:beat_uncoded}(Coded linear inverse beats uncoded)
Suppose the i.i.d. Assumption \ref{ass:general} and \ref{ass:time_iid} hold and suppose $\Gbf$ is a $k\times n$ submatrix of an $n\times n$ Fourier transform matrix $\Fbf$. Then, expected error of the coded linear inverse is strictly less than that of uncoded:
\begin{equation}
\Ep_f \left[\twonorm{\Ebf_{\text{uncoded}} }^2\right]-\Ep_f \left[\twonorm{\Ebf_{\text{coded}} }^2\right]\ge \Ep_f[\trace(\J_2\J_4^{-1}\J_2^\top)],
\end{equation}
where $\J_2$ and $\J_4$ are defined in \eqref{eqn:block_form}.
\end{theorem}
\begin{proof}
From Corollary \ref{cor:upper_bound}, for fixed $l_i,1\le i\le n$,
\begin{equation}
\Ep[\twonorm{\Ebf_\text{coded}}^2|\mathbf{l}]\le\trace(\J_1)-\trace(\J_2\J_4^{-1}\J_2^\top).
\end{equation}
We will show that
\begin{equation}\label{eqn:error_uncoded_exp}
\Ep_f[\trace(\J_1)]=\Ep_f \left[\twonorm{\Ebf_{\text{uncoded}} }^2\right],
\end{equation}
which completes the proof. To show \eqref{eqn:error_uncoded_exp}, first note that from \eqref{eqn:error_uncoded_f},
\begin{equation}\label{eqn:der3}
\Ep_f\left[\twonorm{\Ebf_\text{uncoded}}^2\right]= k\Ep_f[\trace(\C(l_1))].
\end{equation}
Since $\Gbf:=[g_{j,i}]$ is a submatrix of a Fourier matrix, we have $|g_{ji}|^2=1/n$. Thus, $\J_1=\Gbf\boldsymbol{\Lambda}\Gbf^\top$ satisfies\vspace{-3mm}
\[\trace(\J_1)=\sum_{j=1}^k \sum_{i=1}^n |g_{ji}|^2 \trace(\C(l_i))=\frac{k}{n}\sum_{i=1}^n \trace(\C(l_i)).\]
Therefore\vspace{-3mm},
\begin{equation}
\Ep_f[\trace(\J_1)]=k\Ep_f[\trace(\C(l_1))].
\end{equation}
which, along with \eqref{eqn:der3}, completes the proof of \eqref{eqn:error_uncoded_exp}, and hence also the proof of Theorem~\ref{thm:beat_uncoded}.
\end{proof}

\subsection{Comparison between the replication-based and coded linear inverse before a deadline}\label{sec:comp_with_rep}

Consider an alternative way of doing linear inverse using $n>k$ workers. In this paper, we only consider the case when $n-k<k$, i.e., the number of extra workers is only slightly bigger than the number of problems (both in theory and in experiments). Since we have $n-k$ extra workers, a natural way is to pick any $(n-k)$ linear inverse problems and replicate them using these extra $(n-k)$ workers. After we obtain two computation results for the same equation, we use two natural ``decoding'' strategies for this replication-based linear inverse: (i) choose the worker with higher number of iterations; (ii) compute the weighted average using weights $\frac{w_1}{w_1+w_2}$ and $\frac{w_2}{w_1+w_2}$, where $w_1=1/\sqrt{\trace(\C(l_1))}$ and $w_2=1/\sqrt{\trace(\C(l_2))}$, and $l_1$ and $l_2$ are the number of iterations completed at the two workers.


\begin{theorem}
The replication-based schemes satisfies the following lower bound on the mean-squared error:
\begin{equation}\label{eqn:error_replication}
\begin{split}
\Ep_f\left[\twonorm{\Ebf_\text{rep}}^2\right]>&\Ep_f\left[\twonorm{\Ebf_\text{uncoded}}^2\right]\\
&-(n-k)\Ep_f[\trace(\C(l_1))].
\end{split}
\end{equation}
\end{theorem}

\begin{proof}[Proof overview]
Here the goal is to obtain a lower bound on the MSE of replication-based linear inverse and compare it with an upper bound on the MSE of coded linear inverse.

Note that if an extra worker is used to replicate the computation at the $i$-th worker, i.e., the linear inverse problem with input $\rbf_i$ is solved on two workers, the expected error of the result of the $i$-th problem could at best reduced from $\Ep_f[\trace(\C(l_1))]$ to zero\footnote{Although this is clearly a loose bound, it makes for convenient comparison with coded linear inverse}. Therefore, $(n-k)$ extra workers make the error decrease by at most (and strictly smaller than) $ (n-k) \Ep_f[\trace(\C(l_1))]$.
\end{proof}

Using this lower bound, we can provably show that coded linear inverse beats replication-based linear inverse when certain conditions are satisfied. One crucial condition is that the distribution of the random variable $\trace(\C(l))$ satisfies a ``variance heavy-tail'' property defined as follows.
\begin{definition}
The random variable $\trace(\C(l))$ is said to have a ``$\rho$-variance heavy-tail'' property if
\begin{equation}
\text{var}_f[\trace(\C(l))]>\rho\Ep_f^2[\trace(\C(l))],
\end{equation}
\end{definition}
for some constant $\rho>1$. For the coded linear inverse, we will use a Fourier code the generator matrix $\Gbf$ of which is a submatrix of a Fourier matrix. This particular choice of code is only for ease of analysis in comparing coded linear inverse and replication-based linear inverse. In practice, the code that minimizes  mean-squared error should be chosen.

\begin{theorem}\label{thm:beat_replication}(Coded linear inverse beats replication)
Suppose the i.i.d. Assumption \ref{ass:general} and \ref{ass:time_iid} hold and $\Gbf$ is a $k\times n$ submatrix of an $n\times n$ Fourier matrix $\Fbf$. Further, suppose $(n-k)=o(\sqrt{n})$. Then, the expected error of the coded linear inverse satisfies
\begin{equation}\label{eqn:coded_and_replication}
\begin{split}
\lim_{n\to\infty}\frac{1}{n-k}\left[\Ep_f\left[\twonorm{\Ebf_\text{uncoded}}^2\right]-\Ep_f\left[\twonorm{\Ebf_\text{coded}}^2\right]\right]\\
\ge \frac{\text{var}_f[\trace(\C(l_1))]}{\Ep_f[\trace(\C(l_1))]}.
\end{split}
\end{equation}
Moreover, if the random variable $\trace(\C(l))$ satisfies the $\rho$-variance heavy-tail property for $\rho>1$, coded linear inverse outperforms replication-based linear inverse in the following sense,
\begin{equation}\label{eqn:coded_beats_replication}
\begin{split}
&\lim_{n\to\infty}\frac{1}{(n-k)}\left[\Ep_f\left[\twonorm{\Ebf_\text{uncoded}}^2\right]-\Ep_f\left[\twonorm{\Ebf_\text{rep}}^2\right]\right]\\
< \frac{1}{\rho}&\lim_{n\to\infty}\frac{1}{(n-k)}\left[\Ep_f\left[\twonorm{\Ebf_\text{uncoded}}^2\right]-\Ep_f\left[\twonorm{\Ebf_\text{coded}}^2\right]\right].
\end{split}
\end{equation}
\end{theorem}
\begin{proof}[Proof overview]
See appendix Section~\ref{app:beatrep} for a complete and rigorous proof. \textbf{Here we only provide the main intuition behind the proof.} From Theorem~\ref{thm:beat_uncoded}, we have
\begin{equation}\label{eqn:trace_der0}
\Ep_f \left[\twonorm{\Ebf_{\text{uncoded}} }^2\right]-\Ep_f \left[\twonorm{\Ebf_{\text{coded}} }^2\right]\ge\Ep_f[\trace(\J_2\J_4^{-1}\J_2^\top)].
\end{equation}
To prove \eqref{eqn:coded_and_replication}, the main technical difficulty is to simplify the term $\trace(\J_2\J_4^{-1}\J_2^\top)$. For a Fourier matrix $\Fbf$, we are able to show that the matrix $\Fbf \boldsymbol{\Lambda}\Fbf^\top=   \left[
\begin{matrix}
   {\J_1} & {\J_2} \\
   {\J_2^\top} & {\J_4} \\
\end{matrix}  \\
\right]$ (see Corollary~\ref{cor:upper_bound}) is a Toeplitz matrix, which provides a good structure for us to study its behavior. Then, we use the Gershgorin circle theorem \cite{golub2012matrix} (with some algebraic derivation) to show that the maximum eigenvalue of $\J_4$ satisfies $\sigma_\text{max}(\J_4)\approx \Ep_f[\trace(\C(l_1))]$, and separately using some algebraic manipulations, we show
\begin{equation}
\trace(\J_2\J_2^\top)\approx (n-k) \text{var}_f[\trace(\C(l_1))],
\end{equation}
for large matrix size $n$. Since
\begin{equation}
\begin{split}
\trace(\J_2\J_4^{-1}\J_2^\top)\ge \trace(\J_2(\sigma_\text{max}(\J_4))^{-1}\J_2^\top)\\
=\frac{1}{\sigma_\text{max}(\J_4)} \trace(\J_2\J_2^\top),
\end{split}
\end{equation}
we obtain
\begin{equation}\label{eqn:trace_der1}
\trace(\J_2\J_4^{-1}\J_2^\top)\ge \frac{(n-k)\text{var}_f[\trace(\C(l_1))]}{\Ep_f[\trace(\C(l_1))]},
\end{equation}
for large $n$. Then, \eqref{eqn:coded_and_replication} can be proved by plugging \eqref{eqn:trace_der1} into \eqref{eqn:trace_der0}. After that, we can combine \eqref{eqn:coded_and_replication}, \eqref{eqn:error_replication} and the variance heavy-tail property to prove \eqref{eqn:coded_beats_replication}.
\end{proof}

\subsection{Asymptotic Comparison between Coded, Uncoded and Replication-based linear inverse as the Deadline \texorpdfstring{$T_\text{dl}\to\infty$}{Lg}}\label{sec:infinite_time_analysis}
Consider the coded and uncoded linear inverse when the overall computation time $T_{dl}\to \infty$. From Theorem~\ref{thm:upper_bound} and Theorem~\ref{thm:uncoded}, the computation error of uncoded and coded linear inverse are respectively
\begin{equation}
\Ep\left[\twonorm{\Ebf_\text{uncoded}}^2|\mathbf{l}\right]=\sum_{i=1}^k \trace\left(\C(l_i)\right),
\end{equation}
\begin{equation}
\Ep[\twonorm{\Ebf_\text{coded}}^2|\mathbf{l}]\le\sigma_\text{max}(\Gbf^\top\Gbf)\trace\left[(\Gbf \boldsymbol{\Lambda}^{-1} \Gbf^\top)^{-1}\right],
\end{equation}
where the matrix $\boldsymbol{\Lambda}$ is
\begin{equation}
\boldsymbol{\Lambda}=\diag\left[\trace(\C(l_1)),\ldots,\trace(\C(l_n))\right],
\end{equation}
and the matrices $\C(l_i),i=1,\ldots,n$ are defined as
\begin{equation}
\C(l_i)=\B^{l_i}\C_E(\B^\top)^{l_i}.
\end{equation}

\begin{assumption}\label{ass:speed_constanl}
We assume the computation time of one power iteration is fixed at each worker for each linear inverse computation, i.e., there exist $n$ random variables $v_1,v_2,\ldots v_n$ such that $l_i=\lceil\frac{T_\text{dl}}{v_i}\rceil,i=1,2,\ldots n$.
\end{assumption}
The above assumption is validated in experiments in appendix Section~\ref{app:speed_constant}.

The $k$-th order statistic of a statistic sample is equal to its $k$-th smallest value. Suppose the order statistics of the sequence $v_1,v_2,\ldots v_n$ are $v_{i_1}<v_{i_2}<\ldots v_{i_n}$, where $\{i_1,i_2,\ldots i_n\}$ is a permutation of $\{1,2,\ldots n\}$. Denote by $[k]$ the set $\{1,2,\ldots k\}$ and $[n]$ the set $\{1,2,\ldots n\}$.

\begin{theorem}\label{thm:TtoInf}
(Error exponent comparison when $T_\text{dl}\to\infty$) Suppose the i.i.d. Assumption \ref{ass:general} and Assumption \ref{ass:speed_constanl} hold. Suppose $n-k<k$. Then, the error exponents of the coded and uncoded computation schemes satisfy
\begin{equation}
\lim_{T_\text{dl}\to\infty,l_i=\lceil\frac{T_\text{dl}}{v_i}\rceil}-\frac{1}{T_\text{dl}}\log \Ep[\twonorm{\Ebf_\text{coded}}^2|\mathbf{l}]\ge \frac{2}{v_{i_k}} \log\frac{1}{1-d},
\end{equation}
\begin{equation}
\begin{split}
&\lim_{T_\text{dl}\to\infty,l_i=\lceil\frac{T_\text{dl}}{v_i}\rceil}-\frac{1}{T_\text{dl}}\log \Ep[\twonorm{\Ebf_\text{uncoded}}^2|\mathbf{l}]\\
&=\lim_{T_\text{dl}\to\infty,l_i=\lceil\frac{T_\text{dl}}{v_i}\rceil}-\frac{1}{T_\text{dl}}\log \Ep[\twonorm{\Ebf_\text{rep}}^2|\mathbf{l}]
=\frac{2}{\max_{i\in [k]} v_i} \log\frac{1}{1-d},
\end{split}
\end{equation}
The error exponents of uncoded, replication and coded linear inverse satisfy coded$>$replication=uncoded.

Here the expectation $\Ep[\cdot|\mathbf{l}]$ is only taken with respect to the randomness of the linear inverse sequence $\x_i,i=1,2,\ldots k$, and conditioned on the number of iterations $\mathbf{l}$. The limit $\lim_{T_\text{dl}\to\infty}$ is taken under the Assumption~\ref{ass:speed_constanl}, i.e., $ l_i=\lceil\frac{T_\text{dl}}{v_i}\rceil$.
\end{theorem}
\begin{proof}[Proof overview]
See appendix Section~\ref{app:error_exponenl} for a detailed proof. The main intuition behind this result is the following: when $T_\text{dl}$ approaches infinity, the error of uncoded computation is dominated by the slowest worker among the first $k$ workers, which has per-iteration time $\max_{i\in [k]} v_i$. For the replication-based scheme, since the number of extra workers $n-k<k$, there is a non-zero probability (which does not change with $T_\text{dl}$) that the $n-k$ extra workers do not replicate the computation in the slowest one among the first worker. Therefore, replication when $n-k<k$ does not improve the error exponent, because the error is dominated by this slowest worker. For coded computation, we show in appendix Section~\ref{app:error_exponenl} that the slowest $n-k$ workers among the overall $n$ workers do not affect the error exponent, which means that the error is dominated by the $k$-th fastest worker, which has per-iteration time $v_{i_k}$. Since the $k$-th fastest worker among all $n$ workers can not be slower than the slowest one among the first (unordered) $k$ workers, the error exponent of coded linear inverse is larger than that of the uncoded and the replication-based linear inverse.
\end{proof}

\section{Analyzing the Computational Complexity}\label{sec:complexity}
\subsection{Encoding and decoding complexity}
We first show that the encoding and decoding complexity of Algorithm~\ref{alg:pg} are in scaling-sense smaller than that of the computation at each worker. This is important to ensure that straggling comes from the parallel workers, not the encoder or decoder. The proof of the following Theorem is in appendix Section \ref{app:complexity}.

\begin{theorem}\label{thm:complexity}
The computational complexity for the encoding and decoding is $\Theta (nkN)$, where $N$ is the number of rows in the matrix $\B$ and $k, n$ depend on the number of available workers assuming that each worker performs a single linear inverse computation. For a general dense matrix $\B$, the computational complexity of computing linear inverse at each worker is $\Theta (N^2 l)$, where $l$ is the number of iterations in the specified iterative algorithm. The complexity of encoding and decoding is smaller than that of the computation at each user for large $\B$ matrices (large $N$).
\end{theorem}

The complexity of encoding and decoding can be further reduced if the Coppersmith-Winograd algorithm for matrix-matrix multiplication is used \cite{coppersmith1990matrix}. In our experiment on the Google Plus graph for computing PageRank, the computation time at each worker is $30$ seconds and the encoding and decoding time at the central controller is about 1 second.

\subsection{Computing the Matrix \texorpdfstring{$\boldsymbol{\Lambda}$}{Lg}}\label{sec:Lambda}
One difficulty in our coded linear inverse algorithm is computing the entries $\trace(\C(l))=\trace\left(\B^l\C_E(\B^\top)^l\right)$ in the weight matrix $\boldsymbol{\Lambda}$ in \eqref{eqn:Gamma}, which involves a number of matrix-matrix multiplications. One way to side-step this problem is to estimate $\trace(\C(l))$ using Monte Carlo simulations. Concretely, choose $m$ i.i.d. $N$-variate random vectors $\a_1,\a_2,\ldots\a_m$ that are distributed the same as the initial error $\e^{(0)}$ after Assumption~\ref{ass:general}. Then, compute the statistic
\begin{equation}
\hat{\gamma}_{m,l}=\frac{1}{m}\sum_{j=1}^m \twonorm{\B^l\a_j}^2,l=1,2,\ldots T_u,
\end{equation}
where $T_u$ is an upper bound of the number of iterations in a practical iterative computing algorithm. The lemma in appendix Section \ref{app:computing_Lambda} shows that $\hat{\gamma}_{m,l}$ is an unbiased and asymptotically consistent estimator of $\trace(\C(l))$ for all $l$. In our experiments on PageRank, for each graph we choose $m= 10$ and estimate $\trace(\C(l))$ before implementing the coded linear inverse algorithm (in this case it is the coded power-iteration algorithm), which has the same complexity as solving $m=10$ extra linear inverse problems.

For the correlated case, we have to compute a slightly modified weighting matrix denoted by $\tilde{\boldsymbol\Lambda}$ in \eqref{eqn:tildeLambda}. The only change is that we have to compute $\Psi_{i,j}$ in \eqref{eqn:tildeLambda_ij} for all possible $l_i,l_j$ such that $1\le l_i,l_j\le T_u$. We also choose $m$ i.i.d. $N$-variate random vectors $\mathbf{b}_1,\mathbf{b}_2,\ldots \mathbf{b}_m$ that are distributed with mean $\0_N$ and covariance $\C_\text{cor}$, which is the same as the correlation part according to Assumption~\ref{ass:stationary}. Then, compute the statistic
\begin{equation}
\hat{\gamma}_{m,(l_i,l_j)}=\frac{1}{m}\sum_{u=1}^m \mathbf{b}_u\B^{l_j}\B^{l_i}\mathbf{b}_u,1\le l_i,l_j\le T_u.
\end{equation}
Then, the lemma in appendix Section \ref{app:computing_Lambda} shows that $\hat{\gamma}_{m,(l_i,l_j)}$ is also an unbiased and asymptotically consistent estimator of $\Psi_{i,j}$.

\subsection{Analysis on the cost of communication versus computation}\label{sec:com_vs_comp}

In this work, we focus on optimizing the computation cost. However, what if the computation cost is small compared to the overall cost, including the communication cost? If this is true, optimizing the computation cost is not very useful. In what follows, we show that the computation cost is larger than the communication cost in the scaling-sense.

\begin{theorem}\label{thm:communication_cost}
The ratio between computation and communication at the $i$-th worker is $\text{COST}_\text{computation}/\text{COST}_\text{communication}= \Theta(l_i\bar{d})$ operations per integer, where $l_i$ is the number of iterations at the $i$-th worker, and $\bar{d}$ is the average number of non-zeros in each row of the $\B$ matrix. See appendix Section~\ref{app:communication_cost_proof} for a complete proof.
\end{theorem}

\section{Experiments and Simulations}
\subsection{Experiments on Real Systems}\label{sec:experimenl}

\begin{figure}
  \centering
  \includegraphics[scale=0.42]{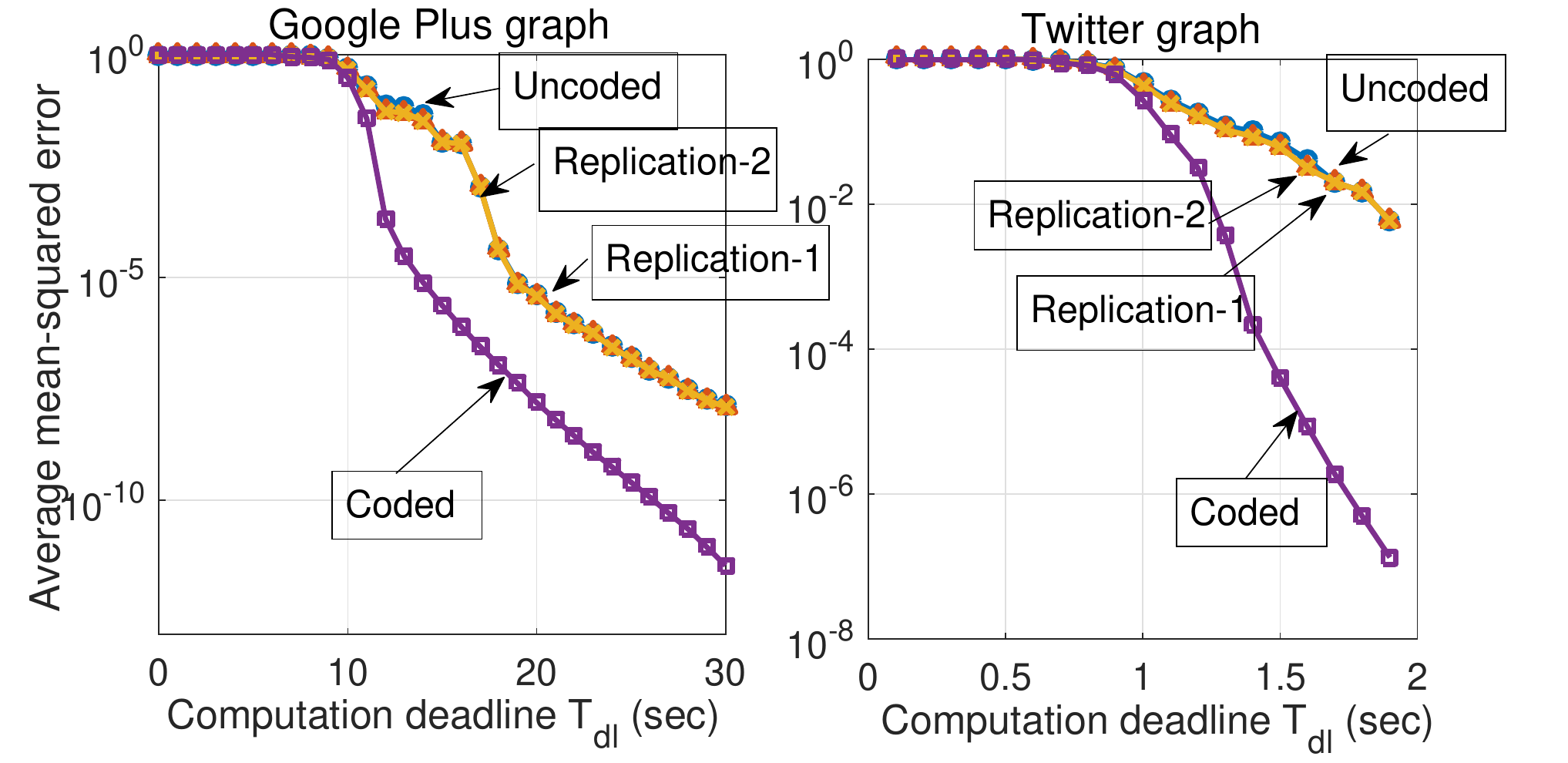}\\
  \caption{Experimentally computed overall mean squared error of uncoded, replication-based and coded personalized PageRank on the Twitter graph and Google Plus graph on a cluseter with 120 workers. The ratio of MSE for repetition-based schemes and coded PageRank increase as $T_{dl}$ increases.}\label{fig:twitter}
\end{figure}

\begin{figure}
  \centering
  \includegraphics[scale=0.45]{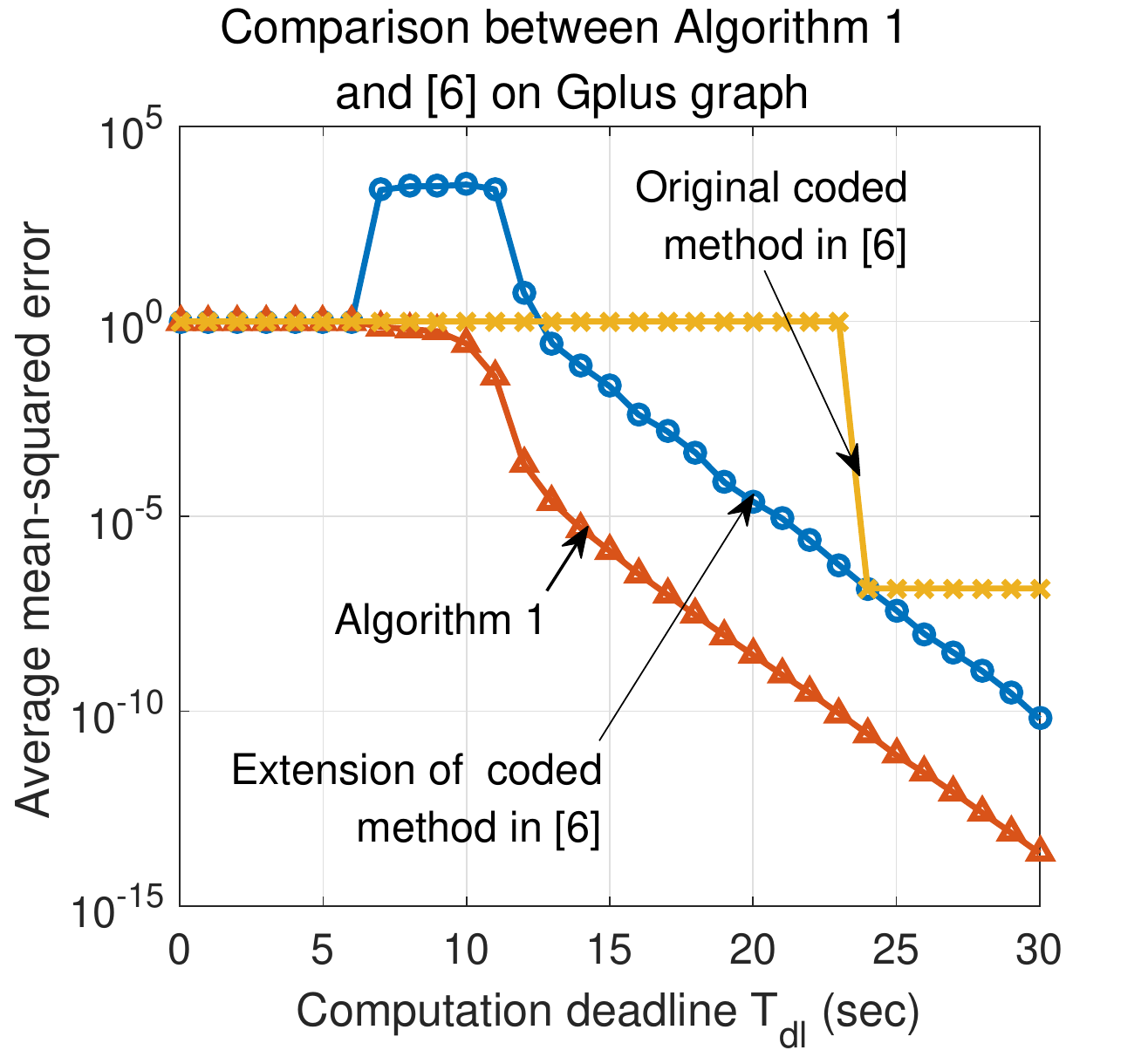}\\
  \caption{Experimental comparison between an extended version of the algorithm in \cite{lee2016speeding} and Algorithm~\ref{alg:pg} on the Google Plus graph. The figure shows that naively extending the general coded computing in \cite{lee2016speeding} using matrix inverse increases the computation error.}\label{fig:lee}
\end{figure}

\begin{figure}
  \centering
  \includegraphics[scale=0.55]{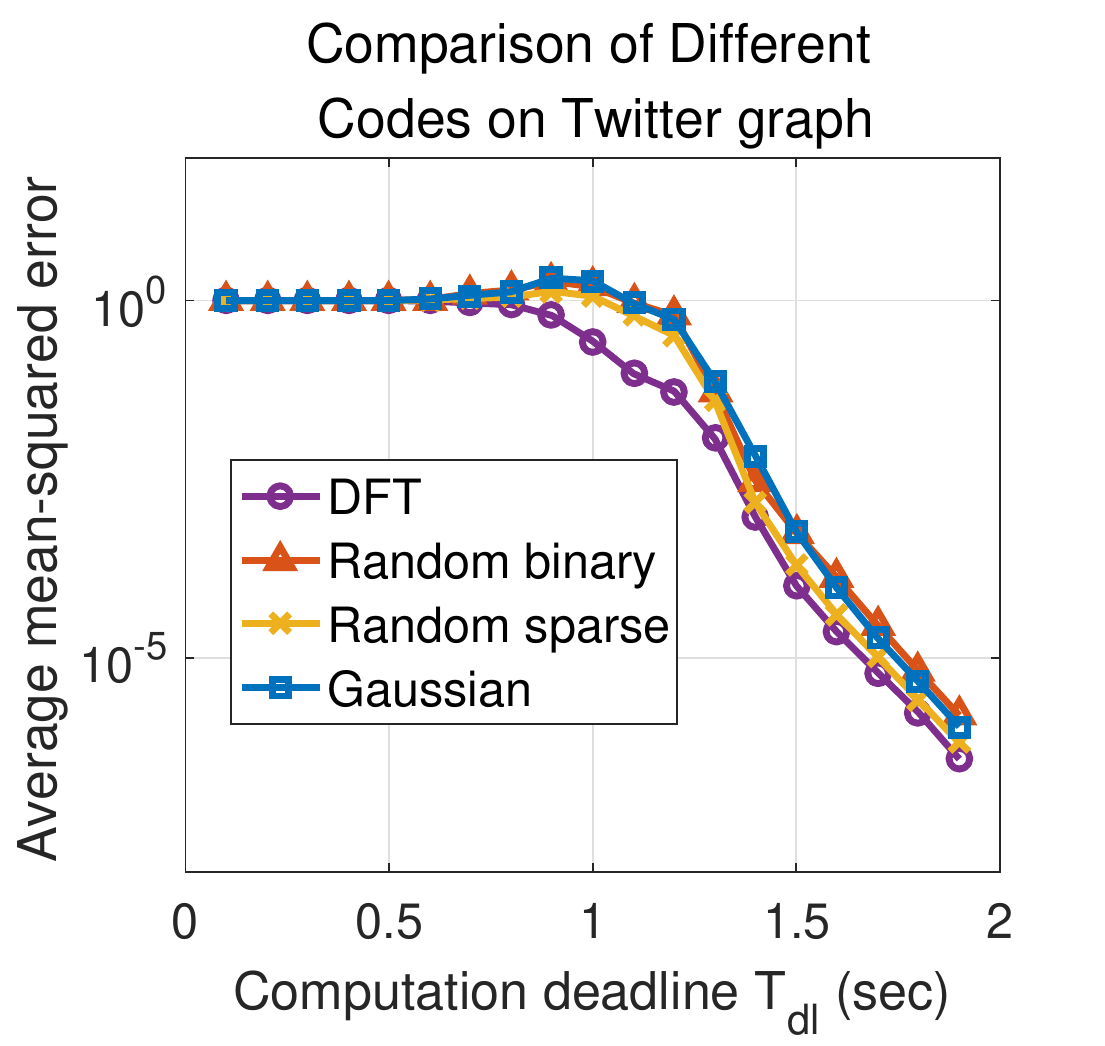}\\
  \caption{Experimental comparison of four different codes on the Twitter graph. In this experiment the DFT-code out-performs the other candidates in mean squared error.}\label{fig:diff_codes}
\end{figure}

We test the performance of the coded linear inverse algorithm for the PageRank problem on the Twitter graph and the Google Plus graph from the SNAP datasets \cite{leskovec2012learning}. The Twitter graph has 81,306 nodes and 1,768,149 edges, and the Google Plus graph has 107,614 nodes and 13,673,453 edges. We use the HT-condor framework in a cluster to conduct the experiments. The task is to solve $k=100$ personalized PageRank problems in parallel using $n=120$ workers. The uncoded algorithm picks the first $k$ workers and uses one worker for each PageRank problem. The two replication-based schemes replicate the computation of the first $n-k$ PageRank problems in the extra $n-k$ workers (see Section~\ref{sec:comp_with_rep}). The coded PageRank uses $n$ workers to solve these $k=100$ equations using Algorithm~\ref{alg:pg}. We use a $(120,100)$ code where the generator matrix is the submatrix composed of the first 100 rows in a $120\times 120$ DFT matrix. The computation results are shown in Fig.~\ref{fig:twitter}. Note that the two graphs of different sizes so the computation in the two experiments takes different time. From Fig.~\ref{fig:twitter}, we can see that the mean-squared error of uncoded and replication-based schemes is larger than that of coded computation by a factor of $10^4$.

We also compare Algorithm~\ref{alg:pg} with the coded computing algorithm proposed in \cite{lee2016speeding}. The original algorithm proposed in \cite{lee2016speeding} is not designed for iterative algorithms, but it has a natural extension to the case of computing before a deadline. Fig. \ref{fig:lee} shows the comparison between the performance of Algorithm~\ref{alg:pg} and this extension of the algorithm from \cite{lee2016speeding}. This extension uses the results from the $k$ fastest workers to retrieve the required PageRank solutions. More concretely, suppose $\S\subset [n]$ is the index set of the $k$ fastest workers. Then, this extension retrieves the solutions to the original $k$ PageRank problems by solving the following equation:
\begin{equation}
\Y_\S=[\x_1^*,\x_2^*,\ldots,\x_k^*]\cdot \Gbf_{\S},
\end{equation}
where $\Y_\S$ is the computation results obtained from the fastest $k$ workers and $\Gbf_{\S}$ is the $k\times k$ submatrix composed of the columns in the generator matrix $\Gbf$ with indexes in $\S$. However, since there is some remaining error at each worker (i.e., the computation results $\Y_\S$ have not converged yet), when conducting the matrix-inverse-based decoding from \cite{lee2016speeding}, the error is magnified due to the large condition number of $\Gbf_{\S}$. This is why the algorithm in \cite{lee2016speeding} cannot be naively applied in the coded PageRank problem.

Finally, we test Algorithm~\ref{alg:pg_stationary} for correlated PageRank queries that are distributed with the stationary covariance matrix in the form of \eqref{eqn:stationary1} and \eqref{eqn:stationary2}. Note that the only change to be made in this case is on the $\boldsymbol\Lambda$ matrix (see equation \eqref{eqn:tildeLambda}). The other settings are exactly the same as the experiments that are shown in Figure~\ref{fig:twitter}. The results on the Twitter social graph are shown in Figure~\ref{fig:diff_codes}. In this case, we also have to compute
\begin{figure}
  \centering
  \includegraphics[scale=0.45]{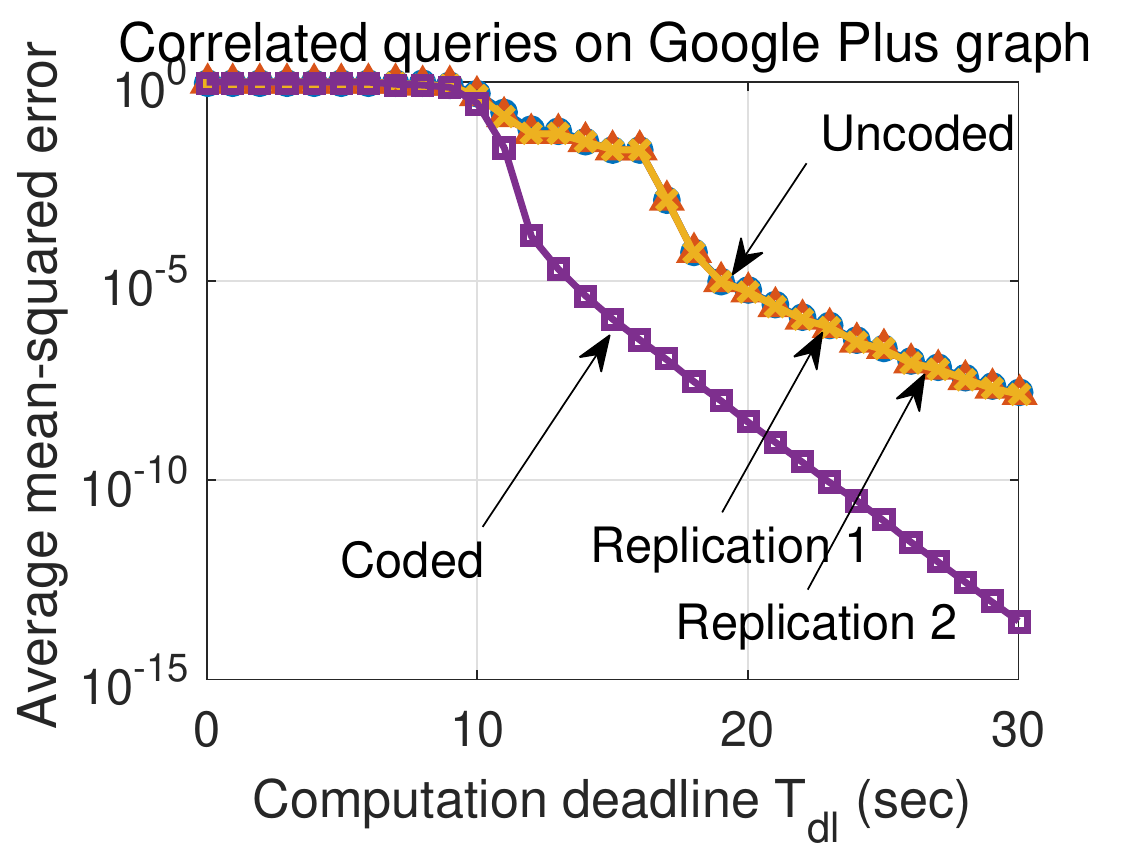}\\
  \caption{Experimentally computed overall mean squared error of uncoded, replication-based and coded personalized PageRank on the Twitter graph on a cluster with 120 workers. The queries are generated using the model from the stationary model in Assumption~\ref{ass:stationary}.}\label{fig:correlated}
\end{figure}

One question remains: what is the best code design for the coded linear inverse algorithm? Although we do not have a concrete answer to this question, we have tested different codes (with different generator matrices $\mathbf{G}$) in the Twitter graph experiment, all using Algorithm~\ref{alg:pg}. The results are shown in Fig.~\ref{fig:lee} (right). The generator matrix used for the ``binary'' curve has i.i.d. binary entries in $\{-1,1\}$. The generator matrix used for the ``sparse'' curve has random binary sparse entries. The generator matrix for the ``Gaussian'' curve has i.i.d. standard Gaussian entries. In this experiment, the DFT-code performs the best. However, finding the best code in general is a meaningful future work.

\subsection{Simulations}\label{app:Simulations}

\begin{figure}
  \centering
  \includegraphics[scale=0.3]{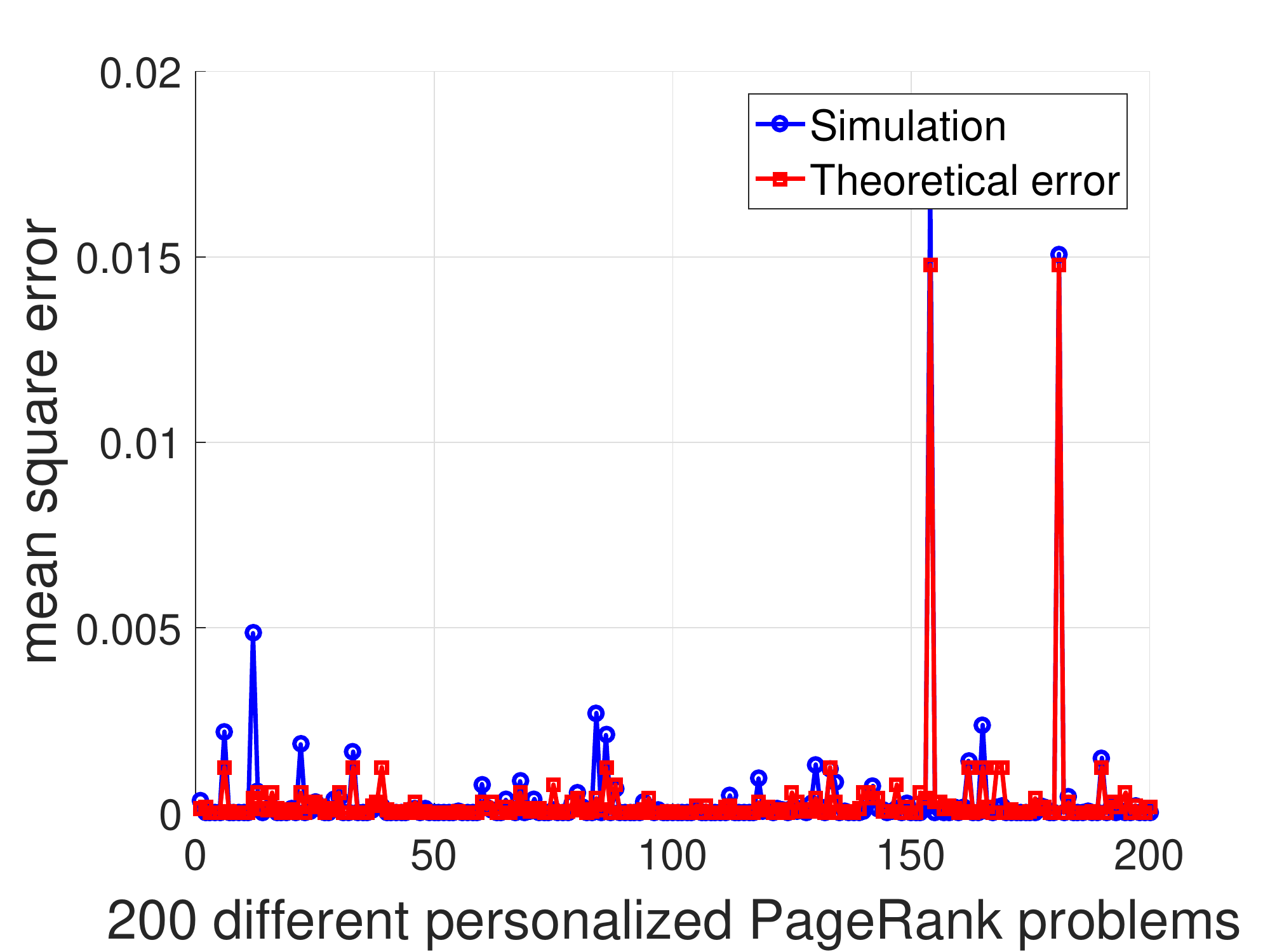}\\
  \caption{This simulation result shows the mean squared error of the computation results for $k=200$ different problems in the uncoded scheme.}\label{fig:error_uncoded}
\end{figure}

\begin{figure}
  \centering
  \includegraphics[scale=0.3]{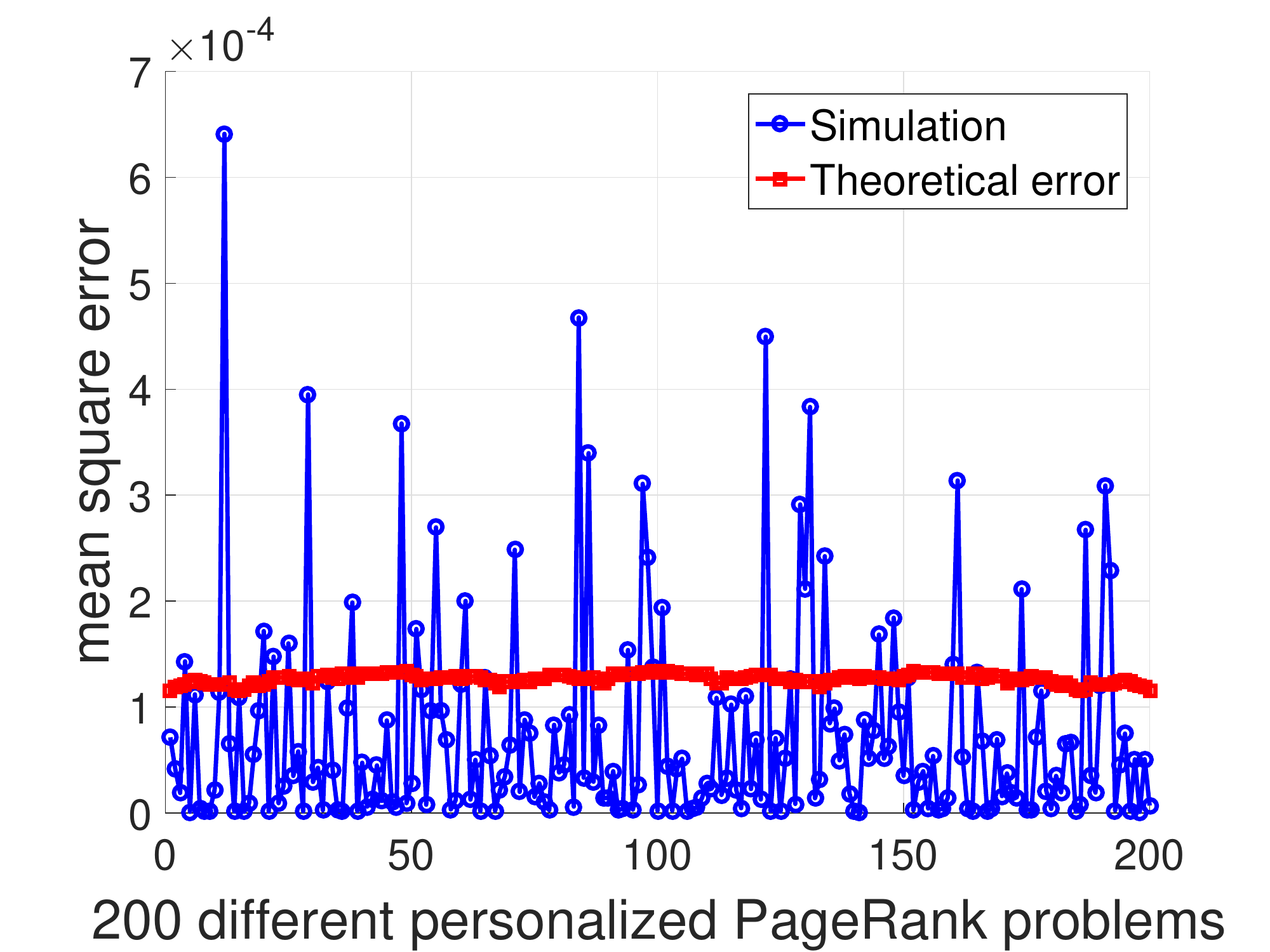}\\
  \caption{This simulation result shows the mean squared error of the computation results for $k=200$ different problems in the coded scheme.}\label{fig:error_coded}
\end{figure}

\begin{figure}
  \centering
  \includegraphics[scale=0.3]{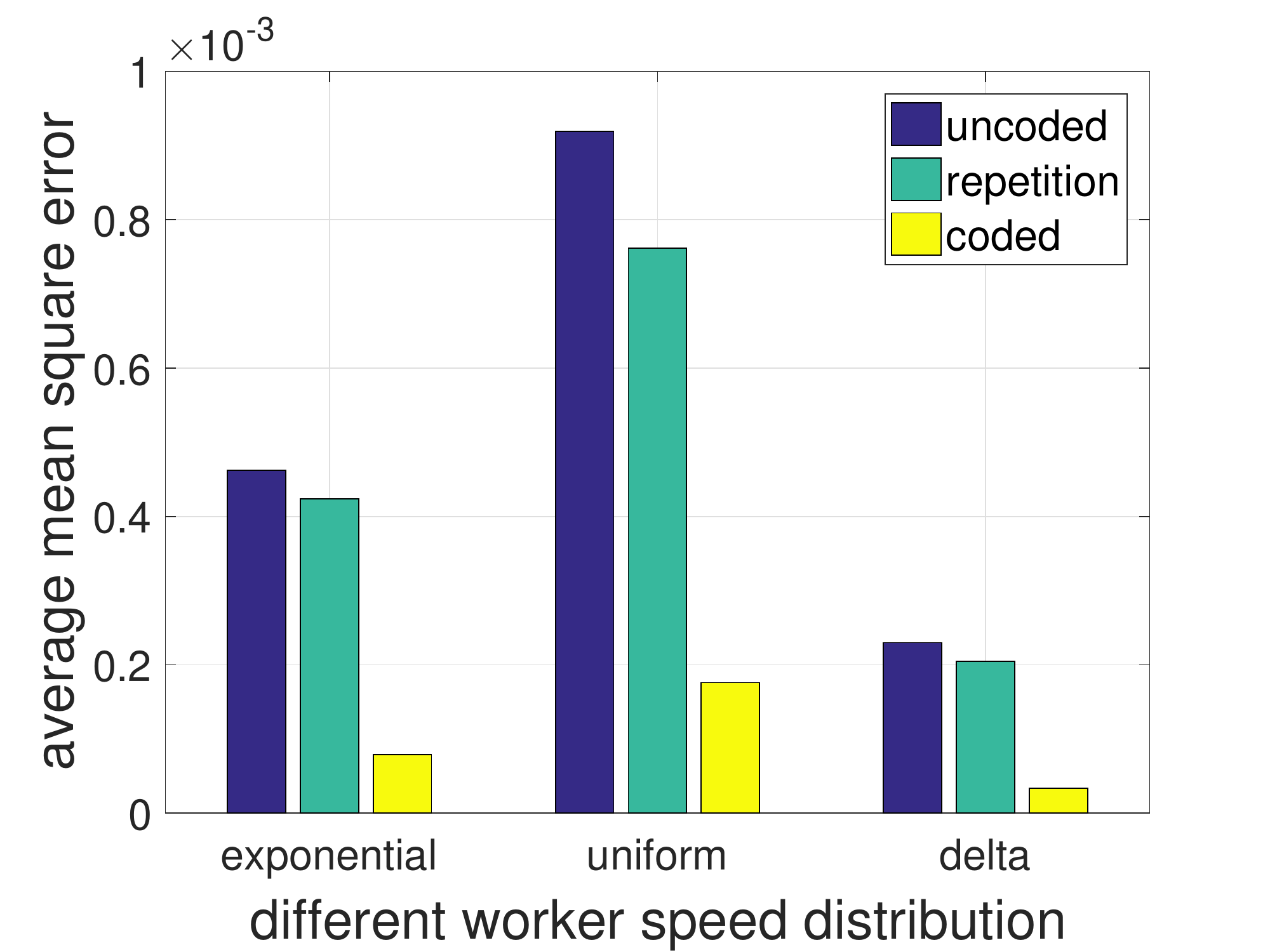}\\
  \caption{This figure shows the mean squared error of uncoded, replication-based and coded PageRank algorithms.}\label{fig:comparison}
\end{figure}

We also test the coded PageRank algorithm in a simulated setup with randomly generated graphs and worker response times. These simulations help us understand looseness in our theoretical bounding techniques. They can also test the performance of the coded Algorithm for different distributions. We simulate Algorithm~\ref{alg:pg} on a randomly generated Erd\"os-R\'enyi graph with $N=500$ nodes and connection probability 0.1. The number of workers $n$ is set to be 240 and the number of PageRank vectors $k$ is set to be 200. We use the first $k=200$ rows of a $240\times 240$ DFT-matrix as the $\Gbf$ matrix in the coded PageRank algorithm in Section \ref{sec:alg}. In Fig. \ref{fig:error_uncoded} and Fig. \ref{fig:error_coded}, we show the simulation result on the mean squared error of all $k=200$ PageRank vectors in both uncoded and coded PageRank, which are respectively shown in Fig.~\ref{fig:error_uncoded} and Fig. \ref{fig:error_coded}. The x-axis represents the computation results for different PageRank problems and the y-axis represents the corresponding mean-squared error. It can be seen that in the uncoded PageRank, some of the PageRank vectors have much higher error than the remaining ones (the blue spikes in Fig. \ref{fig:error_uncoded}), because these are the PageRank vectors returned by the slow workers in the simulation. However, in coded PageRank, the peak-to-average ratio of mean squared error is much lower than in the uncoded PageRank. This means that using coding, we are able to mitigate the straggler effect and achieve more uniform performance across different PageRank computations. From a practical perspective, this means that we can provide fairness to different PageRank queries.

We compare the average mean-squared error of uncoded, replication-based and coded PageRank algorithms in Fig. \ref{fig:comparison}. The first simulation compares these three algorithms when the processing time of one iteration of PageRank computation is exponentially distributed, and the second and third when the number of iterations is uniformly distributed in the range from 1 to 20 and Bernoulli distributed at two points 5 and 20 (which we call ``delta'' distribution). It can be seen that in all three different types of distributions, coded PageRank beats the other two algorithms.

\subsection{Validating Assumption~\ref{ass:speed_constanl} using Experiments}\label{app:speed_constant}
\begin{figure}
  \centering
  \includegraphics[scale=0.35]{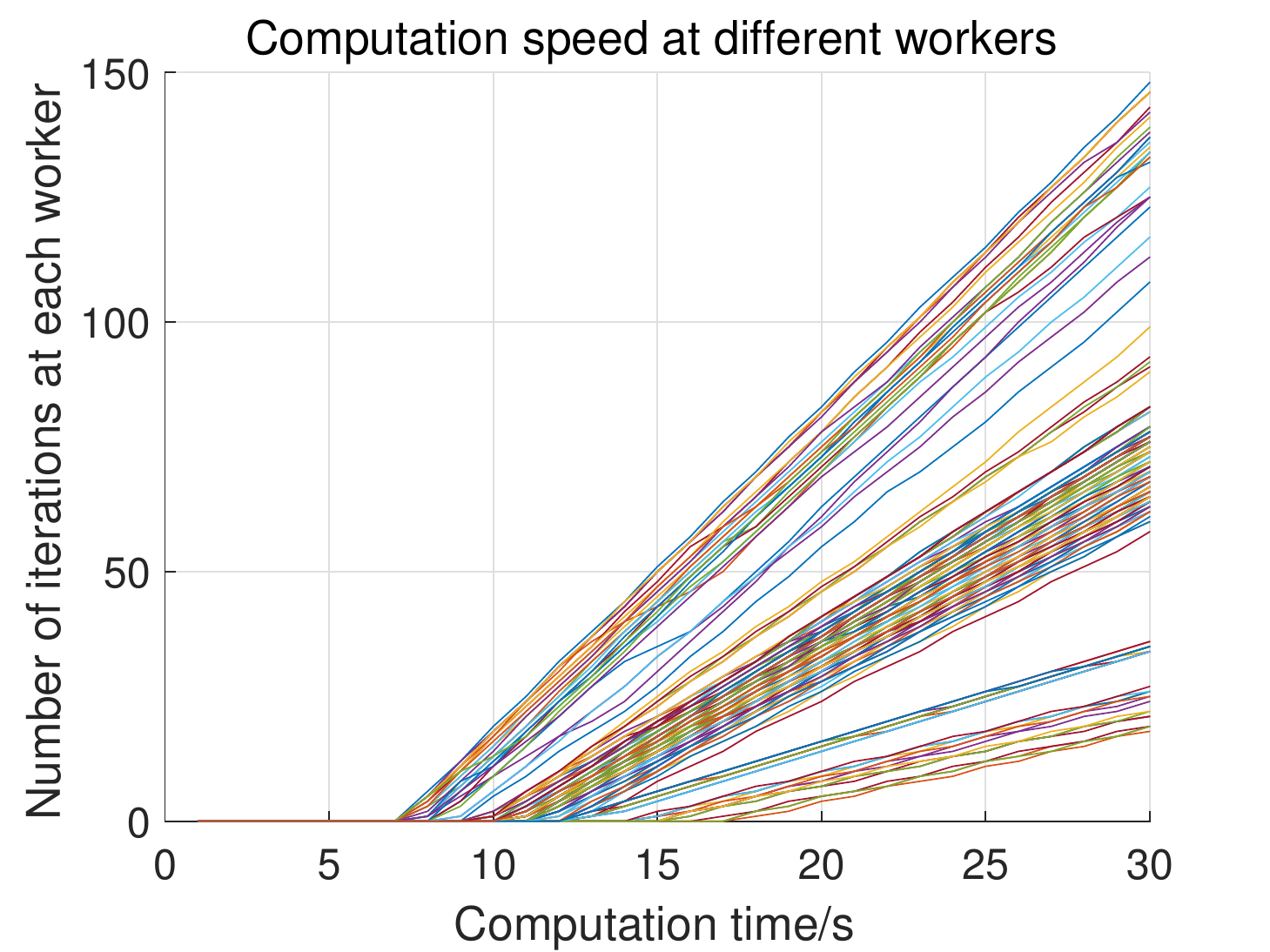}\\
  \caption{This figure shows the number of PageRank power iterations completed at different workers in 30 seconds in the Google Plus experiment.}\label{fig:various_speed}
\end{figure}

Here we provide an experiment that validates Assumption~\ref{ass:speed_constanl} in Section~\ref{sec:infinite_time_analysis}, i.e., the computation time of one power-iteration at the same worker is a constant over time. In Fig.~\ref{fig:various_speed}, we plot the number of power-iterations completed at different workers versus computation time. We can see that the computation speed is indeed constant, which means that Assumption~\ref{ass:speed_constanl} is valid\footnote{In this work, we assume that the statistics of speed distributions are unknown. However, from Fig.~\ref{fig:various_speed}, it may seem that the speeds at different workers are quite predictable. In fact, each time when scheduling tasks to the pool of parallel workers, the central controller assign the tasks through virtual machines instead of actual physical addresses. Therefore, the machines assigned to the same task can be different, and this assignment is transparent to the end-users. Thus, the statistics of speed distributions are generally unobtainable.}. Note that there is a non-zero time cost for loading the graph at each worker. This amount of time does not exist if the network graph is already loaded in the cache of distributed workers for online queries.

\section{Conclusions}

Coded computing for distributed machine learning and data processing systems affected by the straggler effect has become an active research area in recent years \cite{huang2012codes,lee2016speeding,tandon2016gradient,dutta2016short}. By studying coding for iterative algorithms designed for distributed inverse problems, we aim to introduce new applications and analytical tools to the problem of coded computing with stragglers. Since these iterative algorithms designed for inverse problems commonly have decreasing error with time, the partial computation results at stragglers can provide useful information for the distributed computing of the final outputs. By incorporating these partial results, we expect to have improved error reduction compared to treating the results as erasures. Note that this is unlike recent works on coding for multi-stage computing problems \cite{Yang_TIT_17,yang2016coding,yang2016fault}, where the computation error can accumulate with time and coding has to be applied repeatedly to suppress this error accumulation. Our next goal is to study the problem of designing and applying coding techniques to distributed computing under different situations of error accumulation/decay, especially in iterative and multi-stage computations. The exponential decay of error considered in this paper is a special case of the general problem mentioned above.

\section{Appendices}


\subsection{Proof of Theorem~\ref{thm:upper_bound}}\label{app:upper_bound}
We first introduce some notation and preliminary properties that we will use in this proof. Denote by $\text{vec}(\A)$ the vector that is composed of the concatenation of all columns in a matrix $\A$. For example, the vectorization of $A=\left[\begin{matrix}
1 & 2 & 3\\
4 & 5 & 6
\end{matrix}\right]$ is the column vector $\text{vec}(\A)=[1,4,2,5,3,6]^\top$. We will also use the Kronecker product defined as
\begin{equation}
\A_{m\times n}\otimes\B =\left[\begin{matrix}
a_{11}\B & a_{12}\B & \ldots  & a_{1n}\B\\
\vdots   &  \vdots  &  \ddots & \vdots \\
a_{m1}\B & a_{m2}\B & \ldots  & a_{mn}\B.
\end{matrix}\right]
\end{equation}
We now state some properties of the vectorization and Kronecker product.
\begin{lemma}\label{lmm_vec1}
Property 1: if $\A=\B\C$, then
\begin{equation}
\text{vec}(\A)=(\C\otimes \I_N)\text{vec}(\B).
\end{equation}
Property 2: vectorization does not change the Frobenius norm, i.e.,
\begin{equation}
\twonorm{\A}=\twonorm{\text{vec}(\A)}.
\end{equation}
Property 3: The following mixed-product property holds
\begin{equation}
(\A\otimes \B)(\C\otimes \D)=(\A\cdot\C)\otimes(\B\cdot\D),
\end{equation}
if one can form the matrices $\A\C$ and $\B\D$.\\
\noindent Property 4: If $\A$ and $\B$ are both positive semi-definite, $\A\otimes \B$ is also positive semi-definite.\\
\noindent Property 5: Suppose $\C$ is positive semi-definite and $\A\preceq \B$. Then,
\begin{equation}
\A\otimes\C\preceq \B\otimes \C.
\end{equation}
Property 6: (commutative property) Suppose $\A_{m\times n}$ and $\B_{p\times q}$ are two matrices. Then,
\begin{equation}
(\A_{m\times n}\otimes \I_p)\cdot(\I_n\otimes\B_{p\times q})=(\I_m\otimes \B_{p\times q})\cdot (\A_{m\times n}\otimes \I_q).
\end{equation}
Property 7: Suppose $\A$ is an $nN\times nN$ matrix that can be written as
\begin{equation}
\A_{nN\times nN}=\left[\begin{matrix}
\A_{11}&\A_{12}&\ldots &\A_{1n}\\
\A_{21}&\A_{22}&\ldots &\A_{2n}\\
\vdots & \vdots& \ddots & \vdots\\
\A_{n1}&\A_{n2}&\ldots &\A_{nn}
\end{matrix}\right],
\end{equation}
where each $\A_{ij}$ is a square matrix of size $N\times N$. Then, for an arbitrary matrix $\Lbf$ of size $k\times n$,
\begin{equation}
\begin{split}
&\trace\left[(\Lbf\otimes \I_N)\cdot \A\cdot (\Lbf\otimes \I_N)^\top\right]\\
=&\trace\left[\Lbf\cdot\left[\begin{matrix}
\trace[\A_{11}]&\ldots &\trace[\A_{1n}]\\
\vdots & \ddots & \vdots\\
\trace[\A_{n1}]&\ldots &\trace[\A_{nn}]
\end{matrix}\right]\cdot\Lbf^\top\right].
\end{split}
\end{equation}
\end{lemma}
\begin{proof}
See appendix Section~\ref{app:vec}.
\end{proof}

\subsubsection{Computing the explicit form of the error matrix \texorpdfstring{$\Ebf$}{Lg}}
From \eqref{eqn:code_r}, we have encoded the input $\rbf_i$ to the linear inverse problem in the following way:
\begin{equation}\label{eqn:unbiased_der1}
[\s_1,\s_2,\ldots,\s_n]=[\rbf_1,\rbf_2,\ldots,\rbf_k]\cdot\Gbf.
\end{equation}
Since $\x_i^*$ is the solution to the linear inverse problem, we have
\begin{equation}\label{eqn:xnst}
\x_i^*=\C_\text{inv} \rbf_i,
\end{equation}
where $\C_\text{inv}$ is either the direct inverse in \eqref{eqn:general_inverse_solution} for square linear inverse problems or the least-square solution in \eqref{eqn:general_l2_solution} for non-square inverse problems. Define $\y_i^*$ as the solution of the inverse problem with the encoded input $\s_i$. Then, we also have
\begin{equation}\label{eqn:ynst}
\y_i^*=\C_\text{inv} \s_i.
\end{equation}
Left-multiplying $\C_\text{inv}$ on both LHS and RHS of \eqref{eqn:unbiased_der1} and plugging in \eqref{eqn:xnst} and \eqref{eqn:ynst}, we have
\begin{equation}\label{eqn:code_star}
[\y_1^*,\y_2^*,\ldots,\y_n^*]=[\x_1^*,\x_2^*,\ldots,\x_k^*]\cdot\Gbf=\X^*\cdot\Gbf.
\end{equation}

Define $\ep_i^{(l)}=\y_i^{(l_i)}-\y_i^*$, which is the remaining error at the $i$-th worker after $l_i$ iterations. From the explicit form \eqref{eqn:error_iter_A} of the remaining error of the executed iterative algorithm, we have
\begin{equation}\label{eqn:decomp}
\y_i^{(l_i)}=\y_i^*+\ep_i^{(l_i)}=\y_i^*+\mathbf{B}^{l_i}\ep_i^{(0)}.
\end{equation}

Therefore,from the definition of $\Y^{(T_\text{dl})}$ (see \eqref{eqn:pg_results}) and equation \eqref{eqn:code_star} and \eqref{eqn:decomp},
\begin{equation}
\begin{split}
\Y^{(T_\text{dl})}=&[\y_1^{(l_1)},\y_2^{(l_2)},\ldots,\y_n^{(l_n)}]\\
=&[\y_1^*,\y_2^*,\ldots,\y_n^*]+[\ep_1^{(l_1)},\ep_2^{(l_2)},\ldots,\ep_n^{(l_n)}]\\
=&\X^*\cdot\Gbf+[\mathbf{B}^{l_1}\ep_1^{(0)},\ldots,\mathbf{B}^{l_n}\ep_n^{(0)}].
\end{split}
\end{equation}
Plugging in \eqref{eqn:weighted_least_square}, we get the explicit form of $\Ebf=\hat\X^\top-\X^*$:
\begin{equation}\label{eqn:error_form}
\small
\begin{split}
&\hat\X^\top=\\
&=(\Gbf \boldsymbol{\Lambda}^{-1} \Gbf^\top)^{-1}\Gbf\boldsymbol{\Lambda}^{-1}(\Y^{(T_\text{dl})})^\top\\
&=(\Gbf \boldsymbol{\Lambda}^{-1} \Gbf^\top)^{-1}\Gbf\boldsymbol{\Lambda}^{-1}\left[\Gbf^\top (\X^*)^\top+[\mathbf{B}^{l_1}\ep_1^{(0)},\ldots,\mathbf{B}^{l_n}\ep_n^{(0)}]^\top \right]\\
&=(\X^*)^\top+(\Gbf \boldsymbol{\Lambda}^{-1} \Gbf^\top)^{-1}\Gbf\boldsymbol{\Lambda}^{-1}\left[\mathbf{B}^{l_1}\ep_1^{(0)},\ldots,\mathbf{B}^{l_n}\ep_n^{(0)}\right]^\top.
\end{split}
\end{equation}
From \eqref{eqn:code_x0}, \eqref{eqn:code_star} and the definition $\ep_i^{(0)}=\y_i^{(0)}-\y_i^*$ and $\e_i^{(l)}=\x_i^{(0)}-\x_i^*$, we have
\begin{equation}\label{eqn:code_e0}
[\ep_1^{(0)},\ep_2^{(0)},\ldots,\ep_n^{(0)}]=[\e_1^{(0)},\e_2^{(0)},\ldots,\e_k^{(0)}]\cdot\bf{G}.
\end{equation}

\subsubsection{Vectorization of the error matrix \texorpdfstring{$\Ebf$}{Lg}}
From property 2 of Lemma~\ref{lmm_vec1}, vectorization does not change the Frobenius norm, so we have
\begin{equation}
\begin{split}
\Ep[\twonorm\Ebf^2|\mathbf{l}]&=\Ep[\twonorm{\text{vec}(\Ebf)}^2|\mathbf{l}]\\
&=\Ep\left[\trace\left(\text{vec}(\Ebf)\text{vec}(\Ebf)^\top\right)|\mathbf{l}\right].
\end{split}
\end{equation}
Therefore, to prove the conclusion of this theorem, i.e., $\Ep[\twonorm\Ebf^2|\mathbf{l}]\le \sigma_\text{max}(\Gbf^\top\Gbf)\trace\left[(\Gbf \boldsymbol{\Lambda}^{-1} \Gbf^\top)^{-1}\right]$, we only need to show
\begin{equation}\label{eqn:key_ineq_2}
\begin{split}
&\Ep\left[\trace\left(\text{vec}(\Ebf)\text{vec}(\Ebf)^\top\right)|\mathbf{l}\right]\\
\le&\sigma_\text{max}(\Gbf^\top\Gbf)\trace\left[(\Gbf \boldsymbol{\Lambda}^{-1} \Gbf^\top)^{-1}\right].
\end{split}
\end{equation}

\subsubsection{Express the mean-squared error using the vectorization form}
Now we prove \eqref{eqn:key_ineq_2}. From \eqref{eqn:error_form}, we have
\begin{equation}\label{eqn:der1}
\Ebf^\top=(\Gbf \boldsymbol{\Lambda}^{-1} \Gbf^\top)^{-1}\Gbf\boldsymbol{\Lambda}^{-1}[\B^{l_1}\ep_1^{(0)},\ldots,\B^{l_n}\ep_n^{(0)}]^\top,
\end{equation}
which is the same as
\begin{equation}\label{eqn:der12}
\Ebf=[\B^{l_1}\ep_1^{(0)},\ldots,\B^{l_n}\ep_n^{(0)}]\cdot[(\Gbf \boldsymbol{\Lambda}^{-1} \Gbf^\top)^{-1}\Gbf\boldsymbol{\Lambda}^{-1}]^\top.
\end{equation}

From property 1 of Lemma~\ref{lmm_vec1}, \eqref{eqn:der12} means
\begin{equation}
\begin{split}
&\text{vec}(\Ebf)\\
=&\left[(\Gbf \boldsymbol{\Lambda}^{-1} \Gbf^\top)^{-1}\Gbf\boldsymbol{\Lambda}^{-1}\otimes \I_N\right]\cdot \text{vec}([\B^{l_1}\ep_1^{(0)},\ldots,\B^{l_n}\ep_n^{(0)}])\\
=&\left[(\Gbf \boldsymbol{\Lambda}^{-1} \Gbf^\top)^{-1}\Gbf\boldsymbol{\Lambda}^{-1}\otimes \I_N\right]\\
&\cdot \diag[\B^{l_1},\ldots,\B^{l_n}]\cdot\text{vec}([\ep_1^{(0)},\ldots,\ep_n^{(0)}]).
\end{split}
\end{equation}
Define
\begin{equation}
\Lbf=(\Gbf \boldsymbol{\Lambda}^{-1} \Gbf^\top)^{-1}\Gbf\boldsymbol{\Lambda}^{-1},
\end{equation}
\begin{equation}\label{eqn:D}
\D=\diag[\B^{l_1},\ldots,\B^{l_n}],
\end{equation}
and
\begin{equation}
\Ebf_0=\text{vec}([\ep_1^{(0)},\ldots,\ep_n^{(0)}]).
\end{equation}
Then,
\begin{equation}
\text{vec}(\Ebf)=(\Lbf\otimes\I_N)\cdot \D\cdot\Ebf_0.
\end{equation}
Therefore,
\begin{equation}\label{eqn:der2}
\begin{split}
&\Ep\left[\trace\left(\text{vec}(\Ebf)\text{vec}(\Ebf)^\top\right)|\mathbf{l}\right]\\
=&\trace\left((\Lbf\otimes\I_N\cdot \D)\Ep[\Ebf_0\Ebf_0^\top|\mathbf{l}](\Lbf\otimes\I_N\cdot \D)^\top\right).
\end{split}
\end{equation}
\subsubsection{Bounding the term \texorpdfstring{$\Ep[\Ebf_0\Ebf_0^\top|\mathbf{l}]$}{Lg} using the maximum eigenvalue \texorpdfstring{$\sigma_\text{max}(\Gbf^\top\Gbf)$}{Lg}}
Note that $\Ebf_0=\text{vec}([\ep_1^{(0)},\ldots,\ep_n^{(0)}])$. From \eqref{eqn:code_e0}, we have
\begin{equation}
[\ep_1^{(0)},\ep_2^{(0)},\ldots,\ep_n^{(0)}]=[\e_1^{(0)},\e_2^{(0)},\ldots,\e_k^{(0)}]\cdot\bf{G}.
\end{equation}
Therefore, using property 1 of Lemma~\ref{lmm_vec1}, we have
\begin{equation}\label{eqn:e0_vecform}
\Ebf_0=(\Gbf^\top\otimes \I_N)\cdot\text{vec}([\e_1^{(0)},\e_2^{(0)},\ldots,\e_k^{(0)}]).
\end{equation}
From Assumption~\ref{ass:general}, the covariance of $\e_i^{(0)}$ is
\begin{equation}\label{eqn:exp_e0}
\Ep[\e_i^{(0)}(\e_i^{(0)})^\top|\mathbf{l}]=\C_E,i=1,\ldots,k.
\end{equation}
Therefore, from \eqref{eqn:e0_vecform}, we have
\begin{equation}\label{eqn:der4}
\begin{split}
\Ep[\Ebf_0\Ebf_0^\top|\mathbf{l}]=&(\Gbf^\top\otimes \I_N)\cdot\Ep[\text{vec}([\e_1^{(0)},\e_2^{(0)},\ldots,\e_k^{(0)}])\cdot\\
&\quad\text{vec}([\e_1^{(0)},\e_2^{(0)},\ldots,\e_k^{(0)}])^\top|\mathbf{l}]\cdot (\Gbf^\top\otimes \I_N)^\top\\
\overset{(a)}{=}&(\Gbf^\top\otimes \I_N) \cdot (\I_k\otimes\C_E)\cdot (\Gbf^\top\otimes \I_N)^\top\\
=&(\Gbf^\top\otimes \I_N) \cdot (\I_k\otimes\C_E)\cdot (\Gbf\otimes \I_N)\\
\overset{(b)}{=}&(\Gbf^\top\cdot\I_k\cdot\Gbf)\otimes (\I_N\cdot\C_E\cdot\I_N)\\
=&  \Gbf^\top\Gbf \otimes \C_E\\
\overset{(c)}{\preceq}&  \sigma_\text{max}(\Gbf^\top\Gbf)\I_n \otimes \C_E,
\end{split}
\end{equation}
where (a) is from \eqref{eqn:exp_e0}, (b) and (c) follow respectively from property 3 and property 5 of Lemma~\ref{lmm_vec1}.

If $\Gbf$ has orthonormal rows, the eigenvalues of $\Gbf^\top\Gbf$ (which is an $n\times n$ matrix) are all in $(0,1]$. This is why we can remove the term $\sigma_\text{max}(\Gbf^\top\Gbf)$ in \eqref{eqn:error_ub_Gorth} when $\Gbf$ has orthonormal rows. In what follows, we assume $\Gbf$ has orthonormal rows, and the result when $\Gbf$ does not have orthonormal rows follows naturally.

Assuming $\Gbf$ has orthonormal rows, we have
\begin{equation}\label{eqn:der40}
\begin{split}
\Ep[\Ebf_0\Ebf_0^\top|\mathbf{l}]
\preceq \I_n \otimes \C_E.
\end{split}
\end{equation}
Plugging \eqref{eqn:der40} into \eqref{eqn:der2}, we have
\begin{equation}\label{eqn:error_DDT_2}
\begin{split}
&\Ep\left[\trace\left(\text{vec}(\Ebf)\text{vec}(\Ebf)^\top\right)|\mathbf{l}\right]\\
\le&\trace\left((\Lbf\otimes\I_N)\cdot \D(\I_n \otimes \C_E)\D^\top(\Lbf\otimes\I_N)^\top\right),
\end{split}
\end{equation}
where $\D=\diag[\B^{l_1},\ldots,\B^{l_n}]$. Therefore,
\begin{equation}
\begin{split}
&\D(\I_n \otimes \C_E)\D^\top \\
=&\diag[\B^{l_1}\C_E(\B^\top)^{l_1},\ldots,\B^{l_n}\C_E(\B^\top)^{l_n}].
\end{split}
\end{equation}
From the definition of $\C(l_i)$ in \eqref{eqn:CTI},
\begin{equation}\label{eqn:def_DDtop}
\D(\I_n \otimes \C_E)\D^\top =\diag[\C(l_1),\ldots,\C(l_n)].
\end{equation}
\subsubsection{Reducing the dimensionality of \texorpdfstring{$\D(\I_n \otimes \C_E)\D^\top$}{Lg} in the trace expression using property 7 in Lemma~\ref{lmm_vec1}}
From Property 7 in Lemma~\ref{lmm_vec1}, we can simplify \eqref{eqn:error_DDT_200}:
\begin{equation}\label{eqn:error_DDT_200}
\begin{split}
&\Ep\left[\trace\left(\text{vec}(\Ebf)\text{vec}(\Ebf)^\top\right)|\mathbf{l}\right]\\
\le&\trace\left((\Lbf\otimes\I_N)\cdot \D(\I_n \otimes \C_E)\D^\top(\Lbf\otimes\I_N)^\top\right)\\
\overset{(a)}
=&\trace\left((\Lbf\otimes\I_N)\cdot \diag[\C(l_1),\ldots,\C(l_n)](\Lbf\otimes\I_N)^\top\right)\\
\overset{(b)}
{=}&\trace\left[\Lbf\cdot\diag[\trace(\C(l_1)),\ldots,\trace(\C(l_1))]\Lbf^\top\right]\\
\overset{(c)}{=}&\trace\left(\Lbf\boldsymbol{\Lambda}\Lbf^\top\right),
\end{split}
\end{equation}
where (a) is from \eqref{eqn:def_DDtop}, (b) is from Property 7 and (c) is from the definition of $\boldsymbol{\Lambda}$ in \eqref{eqn:Gamma}. Equation \eqref{eqn:error_DDT_200} can be further simplified to
\begin{equation}
\begin{split}
&\Ep\left[\trace\left(\text{vec}(\Ebf)\text{vec}(\Ebf)^\top\right)|\mathbf{l}\right]
\le \trace\left(\Lbf\boldsymbol{\Lambda}\Lbf^\top\right)\\
\overset{(a)}{=}& \trace\left((\Gbf \boldsymbol{\Lambda}^{-1} \Gbf^\top)^{-1}\Gbf\boldsymbol{\Lambda}^{-1}\boldsymbol{\Lambda}((\Gbf \boldsymbol{\Lambda}^{-1} \Gbf^\top)^{-1}\Gbf\boldsymbol{\Lambda}^{-1})^\top \right)\\
=&\trace((\Gbf \boldsymbol{\Lambda}^{-1} \Gbf^\top)^{-1}),
\end{split}
\end{equation}
where (a) is from the definition of the decoding matrix $\Lbf=(\Gbf \boldsymbol{\Lambda}^{-1} \Gbf^\top)^{-1}\Gbf\boldsymbol{\Lambda}^{-1}$. Thus, we have completed the proof of Theorem~\ref{thm:upper_bound} for the case when $\Gbf$ has orthonormal rows. As we argued earlier, the proof when $\Gbf$ does not have orthonormal rows follows immediately (see the text after \eqref{eqn:der4}).
\subsection{Proof of Corollary~\ref{cor:upper_bound}}\label{app:Schur}

First, note that
\begin{equation}
\Gbf=\left[\I_k,\0_{k,n-k}\right]\Fbf.
\end{equation}
Therefore,
\begin{equation}\label{eqn:GGl}
\begin{split}
\Gbf \boldsymbol{\Lambda}^{-1} \Gbf^\top=\left[\I_k,\0_{k,n-k}\right]\Fbf \boldsymbol{\Lambda}^{-1}\Fbf^\top \left[\I_k,\0_{k,n-k}\right]^\top\\
\overset{(a)}{=}\left[\I_k,\0_{k,n-k}\right](\Fbf \boldsymbol{\Lambda}\Fbf^\top)^{-1} \left[\I_k,\0_{k,n-k}\right]^\top,
\end{split}
\end{equation}
where (a) is from $\Fbf^\top \Fbf=\I_n$. Now take the inverse of both sides of \eqref{eqn:block_form}, we have
\begin{equation}\label{eqn:block_form_2}
(\Fbf \boldsymbol{\Lambda}\Fbf^\top)^{-1}=   \left[
\begin{matrix}
   {(\J_1-\J_2\J_4^{-1}\J_2^\top)^{-1}} & {*} \\
   {*} & {*} \\
\end{matrix}  \\
\right]_{n\times n},
\end{equation}
where $*$ is used as a substitute for matrices that are unimportant for our argument. Thus, comparing \eqref{eqn:GGl} and \eqref{eqn:block_form_2},
\begin{equation}
\Gbf \boldsymbol{\Lambda}^{-1} \Gbf^\top=(\J_1-\J_2\J_4^{-1}\J_2^\top)^{-1},
\end{equation}
which means
\begin{equation}\label{eqn:42}
(\Gbf \boldsymbol{\Lambda}^{-1} \Gbf^\top)^{-1}=\J_1-\J_2\J_4^{-1}\J_2^\top.
\end{equation}
From \eqref{eqn:error_ub_Gorth} and \eqref{eqn:42}, the theorem follows.

\subsection{Proof of Theorem \ref{thm:upper_bound_stationary}}\label{app:correlated}

The proof follows the same procedure as the proof of Theorem~\ref{thm:upper_bound}. Basically, we can obtain exactly the same results from \eqref{eqn:unbiased_der1} to \eqref{eqn:e0_vecform} except that all $\boldsymbol\Lambda$ are replaced with $\tilde{\boldsymbol\Lambda}$. However, now that we assume the solutions of the linear inverse problems satisfy \ref{ass:stationary}, we have
\begin{equation}\label{eqn:der400}
\begin{split}
&\Ep[\text{vec}([\e_1^{(0)},\e_2^{(0)},\ldots,\e_k^{(0)}])\text{vec}([\e_1^{(0)},\e_2^{(0)},\ldots,\e_k^{(0)}])^\top|\mathbf{l}]\\
&=\I_k\otimes \C_E+(\one_k\one_k^\top)\otimes \C_\text{cor}.
\end{split}
\end{equation}
Note that the first part $\I_k\otimes \C_E$ is exactly the same as in the proof of Theorem~\ref{thm:upper_bound}, so all conclusions until \eqref{eqn:error_DDT_200} can still be obtained (note that $\sigma_\text{max}(\Gbf^\top\Gbf)$ should be added in the general case) for this part. More specifically, this means \eqref{eqn:error_DDT_200} can be modified to
\begin{equation}\label{eqn:error_DDT_300}
\begin{split}
&\Ep\left[\trace\left(\text{vec}(\Ebf)\text{vec}(\Ebf)^\top\right)|\mathbf{l}\right]\\
\le&\trace\left((\Lbf\otimes\I_N)\cdot \D\boldsymbol\Sigma\D^\top(\Lbf\otimes\I_N)^\top\right)\\
&+\sigma_\text{max}(\Gbf^\top\Gbf)\trace\left(\Lbf\boldsymbol{\Lambda}\Lbf^\top\right),
\end{split}
\end{equation}
where the second term $\sigma_\text{max}(\Gbf^\top\Gbf)\trace\left(\Lbf\boldsymbol{\Lambda}\Lbf^\top\right)$ is the same as in \eqref{eqn:error_DDT_200} because of the first part $\I_k\otimes \C_E$ in \eqref{eqn:der400}. However, the first term $\trace\left((\Lbf\otimes\I_N)\cdot \D\boldsymbol\Sigma\D^\top(\Lbf\otimes\I_N)^\top\right)$ is from the correlation between different inputs, and the matrix $\boldsymbol\Sigma$ is
\begin{equation}\label{eqn:Sigma1}
\boldsymbol\Sigma=(\Gbf^\top\otimes \I_N) \cdot ((\one_k\one_k^\top)\otimes \C_\text{cor})\cdot (\Gbf^\top\otimes \I_N)^\top,
\end{equation}
which is obtained by adding the second term $(\one_k\one_k^\top)\otimes \C_\text{cor}$ in \eqref{eqn:der400} into the step (a) in \eqref{eqn:der4}. From Property 6 of Lemma~\ref{lmm_vec1}, \eqref{eqn:Sigma1} can be simplified to
\begin{equation}
\boldsymbol\Sigma=(\Gbf^\top\one_k\one_k^\top\Gbf)\otimes \C_\text{cor}.
\end{equation}
Therefore, from the definition of \eqref{eqn:D}
\begin{equation}
\begin{split}
&\D \boldsymbol\Sigma\D^\top\\
=&\diag[\B^{l_1},\ldots,\B^{l_n}]\cdot(\Gbf^\top\one_k\one_k^\top\Gbf)\otimes \C_\text{cor}\\
&\cdot\diag[(\B^\top)^{l_1},\ldots,(\B^\top)^{l_n}].
\end{split}
\end{equation}
Define the column vector $\h=\Gbf^\top \one_k:=[h_1,h_2,\ldots h_n]^\top$. Then, $(\Gbf^\top\one_k\one_k^\top\Gbf)\otimes \C_\text{cor}$ can be written as a block matrix where the block on the $i$-th row and the $j$-th column is $h_ih^*_j\C_\text{cor}$. Therefore, After left-multiplying the block diagonal matrix $\diag[\B^{l_1},\ldots,\B^{l_n}]$ and right-multiplying $\diag[(\B^\top)^{l_1},\ldots,(\B^\top)^{l_n}]$, we obtain
\begin{equation}
\D \boldsymbol\Sigma\D^\top=\tilde{\boldsymbol\Psi}=[\tilde{\boldsymbol\Psi}_{i,j}],
\end{equation}
where the block $\tilde{\boldsymbol\Psi}_{i,j}$ on the $i$-th row and the $j$-th column is $h_ih^*_j \B^{l_i}\C_\text{cor}(\B^\top)^{l_j}$. From Property 7 of Lemma~\ref{lmm_vec1}, we have
\begin{equation}\label{eqn:der500}
\begin{split}
&\trace[(\Lbf\otimes\I_N)\cdot \D\boldsymbol\Sigma\D^\top(\Lbf\otimes\I_N)^\top]\\
\overset{(a)}{=}&\trace\left[\Lbf\left[\trace[\tilde{\boldsymbol\Psi}_{i,j}]\right]\Lbf^\top\right]\\
\overset{(b)}{=}&\trace\left[\Lbf\left[h_ih^*_j \trace[\B^{l_i}\C_\text{cor}(\B^\top)^{l_j}]\right]\Lbf^\top\right]\\
\overset{(c)}{=}&\trace\left[\Lbf\diag(\h)\left[ \trace[\B^{l_i}\C_\text{cor}(\B^\top)^{l_j}]\right]\diag(\h^\top)\Lbf^\top\right]\\
\overset{(d)}{=}&\trace\left[\Lbf\diag\{\Gbf^\top \one_k\}\cdot \boldsymbol{\Psi}\cdot \diag\{\Gbf^\top \one_k\}^\top\Lbf^\top\right],
\end{split}
\end{equation}
where step (a) is from Property 7 of Lemma \ref{lmm_vec1} and the notation $\left[\trace[\tilde{\boldsymbol\Psi}_{i,j}]\right]$ means the $n\times n$ matrix with entries $\trace[\tilde{\boldsymbol\Psi}_{i,j}]$, (b) is from the definition of $\tilde{\boldsymbol\Psi}_{i,j}$ below \eqref{eqn:der500}, (c) is from the definition $\h=\Gbf^\top \one_k:=[h_1,h_2,\ldots h_n]^\top$, and (d) is from the definition of the matrix $\boldsymbol\Psi$ in \eqref{eqn:tildeLambda_ij}. Plugging \eqref{eqn:der500} into \eqref{eqn:error_DDT_300}, we obtain
\begin{equation}\label{eqn:error_DDT_600}
\begin{split}
&\Ep\left[\trace\left(\text{vec}(\Ebf)\text{vec}(\Ebf)^\top\right)|\mathbf{l}\right]\\
\le&\trace\left[\Lbf\diag\{\Gbf^\top \one_k\}\cdot \boldsymbol{\Psi}\cdot \diag\{\Gbf^\top \one_k\}^\top\Lbf^\top\right]\\
&+\sigma_\text{max}(\Gbf^\top\Gbf)\trace\left(\Lbf\boldsymbol{\Lambda}\Lbf^\top\right)\\
=&\trace[\Lbf\tilde{\boldsymbol\Lambda}\Lbf^\top],
\end{split}
\end{equation}
where $\tilde{\boldsymbol{\Lambda}}=\sigma_\text{max}(\Gbf^\top\Gbf)\boldsymbol{\Lambda}+\diag\{\Gbf^\top \one_k\}\cdot \boldsymbol{\Psi}\cdot \diag\{\Gbf^\top \one_k\}^\top$, which is the same as in \eqref{eqn:tildeLambda}. Therefore
\begin{equation}
\begin{split}
&\Ep\left[\trace\left(\text{vec}(\Ebf)\text{vec}(\Ebf)^\top\right)|\mathbf{l}\right]
\le \trace\left(\Lbf\tilde{\boldsymbol{\Lambda}}\Lbf^\top\right)\\
\overset{(a)}{=}& \trace\left((\Gbf \tilde{\boldsymbol{\Lambda}}^{-1} \Gbf^\top)^{-1}\Gbf\tilde{\boldsymbol{\Lambda}}^{-1}\tilde{\boldsymbol{\Lambda}}((\Gbf \tilde{\boldsymbol{\Lambda}}^{-1} \Gbf^\top)^{-1}\Gbf\tilde{\boldsymbol{\Lambda}}^{-1})^\top \right)\\
=&\trace((\Gbf \tilde{\boldsymbol{\Lambda}}^{-1} \Gbf^\top)^{-1}),
\end{split}
\end{equation}
where (a) is from the definition of the decoding matrix $\Lbf=(\Gbf \tilde{\boldsymbol{\Lambda}}^{-1} \Gbf^\top)^{-1}\Gbf\tilde{\boldsymbol{\Lambda}}^{-1}$.

\subsection{Proof of Theorem~\ref{thm:uncoded}}\label{app:uncoded_replication}
In this section, we compute the residual error of the uncoded linear inverse algorithm. From \eqref{eqn:error_iter}, we have
\begin{equation}\label{eqn:eiexplicit}
\e_i^{(l+1)}=\B\e_i^{(l)}.
\end{equation}
Therefore, in the uncoded scheme, the overall error is
\begin{equation}\label{eqn:error_uncoded_generalized}
\begin{split}
&\Ep\left[\twonorm{\Ebf_{\text{uncoded}} }^2|\textbf{l}\right]\\
=&\Ep\left[\twonorm{[{\e_1^{(l_1)}},{\e_2^{(l_2)}}\ldots,{\e_k^{(l_k)}}]}^2|\mathbf{l}\right]\\
=&\sum_{i=1}^k \Ep\left[\twonorm{[\e_i^{(l_i)}]}^2|\mathbf{l}\right]\\
=& \sum_{i=1}^k \trace\left(\Ep\left[\e_i^{(l_i)}(\e_i^{(l_i)})^\top|\mathbf{l}\right]\right)\\
\overset{(a)}{=}& \sum_{i=1}^k \trace\left(\Ep\left[\B^{l_i}\e_i^{(0)}(\B^{l_i}\e_i^{(0)})^\top|\mathbf{l}\right]\right)\\
=& \sum_{i=1}^k \trace\left(\B^{l_i}\Ep\left[\e_i^{(0)}(\e_i^{(0)})^\top|\mathbf{l}\right](\B^{l_i})^\top\right)\\
=& \sum_{i=1}^k \trace\left(\B^{l_i}\cdot \C_E \cdot(\B^{l_i})^\top\right)\\
=& \sum_{i=1}^k \trace\left(\B^{l_i} \C_E(\B^{l_i})^\top\right)\\
\overset{(b)}{=}&\sum_{i=1}^k \trace\left(\C(l_i)\right),
\end{split}
\end{equation}
where (a) is from \eqref{eqn:eiexplicit} and (b) is from the definition of $\C(l_i)$ in \eqref{eqn:CTI}. Thus, we have proved \eqref{eqn:error_uncoded}. To prove \eqref{eqn:error_uncoded_f}, we note that from the i.i.d. assumption of $l_i$,
\begin{equation}
\begin{split}
\Ep_f \left[\twonorm{\Ebf_{\text{uncoded}} }^2\right]=& \Ep_f\left[\sum_{i=1}^k \trace\left(\C(l_i)\right)\right]\\
=& k\Ep_f[\trace(\C(l_1))].
\end{split}
\end{equation}

\subsection{Proof of Theorem~\ref{thm:beat_replication}}\label{app:beatrep}

From Theorem~\ref{thm:beat_uncoded},
\begin{equation}\label{eqn:der_temp_1}
\Ep_f \left[\twonorm{\Ebf_{\text{uncoded}} }^2\right]-\Ep_f \left[\twonorm{\Ebf_{\text{coded}} }^2\right]\ge \Ep_f[\trace(\J_2\J_4^{-1}\J_2^\top)].
\end{equation}
We now argue that to show \eqref{eqn:coded_and_replication}, we only need to show
\begin{equation}\label{eqn:key_resull}
\begin{split}
&\lim_{n\to\infty}\frac{1}{n-k}\Ep_f[\trace(\J_2\J_4^{-1}\J_2^\top)]
\ge \frac{\text{var}_f[\trace(\C(l_1))]}{\Ep_f[\trace(\C(l_1))]},
\end{split}
\end{equation}
because then, we have
\begin{equation}
\begin{split}
&\lim_{n\to\infty}\frac{1}{(n-k)}\left[\Ep_f\left[\twonorm{\Ebf_\text{uncoded}}^2\right]-\Ep_f\left[\twonorm{\Ebf_\text{coded}}^2\right]\right]\\
\overset{(a)}{\ge}& \lim_{n\to\infty}\frac{1}{(n-k)} \Ep_f[\trace(\J_2\J_4^{-1}\J_2^\top)]\\
\overset{(b)}{\ge} &\frac{\text{var}_f[\trace(\C(l_1))]}{\Ep_f[\trace(\C(l_1))]},
\end{split}
\end{equation}
where (a) follows from \eqref{eqn:der_temp_1} and (b) follows from \eqref{eqn:key_resull}.

Also note that after we prove \eqref{eqn:coded_and_replication}, then using \eqref{eqn:error_replication}, we have
\begin{equation}
\begin{split}
&\Ep_f\left[\twonorm{\Ebf_\text{uncoded}}^2\right]-\Ep_f\left[\twonorm{\Ebf_\text{rep}}^2\right]\\
\le & (n-k)\Ep_f[\trace(\C(l_1))],
\end{split}
\end{equation}
so we have
\begin{equation}
\begin{split}
&\lim_{n\to\infty}\frac{1}{(n-k)}\left[\Ep_f\left[\twonorm{\Ebf_\text{uncoded}}^2\right]-\Ep_f\left[\twonorm{\Ebf_\text{rep}}^2\right]\right]\\
\le & \Ep_f[\trace(\C(l_1))]\\
\overset{(a)}{\le} & \frac{1}{\rho}  \frac{\text{var}_f[\trace(\C(l_1))]}{\Ep_f[\trace(\C(l_1))]}\\
\le &\frac{1}{\rho}\lim_{n\to\infty}\frac{1}{(n-k)}\left[\Ep_f\left[\twonorm{\Ebf_\text{uncoded}}^2\right]-\Ep_f\left[\twonorm{\Ebf_\text{coded}}^2\right]\right],
\end{split}
\end{equation}
which means coded computation beats uncoded computation. Note that step (a) holds because of the variance heavy-tail property.

Therefore, we only need to prove \eqref{eqn:key_resull}. The proof of \eqref{eqn:key_resull} is divided into two steps, and intuition behind each step is provided along the proof. The main intuition is that the Fourier structure of the matrix $\Fbf$ makes the matrix $\J_4$ concentrates around its mean value, which makes the most tricky term $\Ep_f[\trace(\J_2\J_4^{-1}\J_2^\top)]$ analyzable.

\subsubsection{Exploiting the Fourier structure to obtain a Toeplitz covariance matrix}
First, we claim that when $\Fbf_{n\times n}$ is the Fourier transform matrix, the matrix $\Fbf \boldsymbol{\Lambda}\Fbf^\top$ in \eqref{eqn:block_form}
\begin{equation}
\Fbf \boldsymbol{\Lambda}\Fbf^\top=   \left[
\begin{matrix}
   {\J_1} & {\J_2} \\
   {\J_2^\top} & {\J_4} \\
\end{matrix}  \\
\right]_{n\times n},
\end{equation}
is a Toeplitz matrix composed of the Fourier coefficients of the sequence (vector) $s=[\trace(\C(l_1)),\ldots,\trace(\C(l_n))]$. In what follows, we use the simplified notation
\begin{equation}
s_j:=\trace(\C(l_{j+1})),j=0,1,\ldots,n-1.
\end{equation}

\begin{lemma}\label{lmm:Fourier}
If
\begin{equation}
\Fbf = \left(\frac{w^{pq}}{\sqrt{n}}\right)_{p,q=0,1,\ldots,n-1},
\end{equation}
where $w=\exp(-2\pi i/n)$, then
\begin{equation}
\Fbf \boldsymbol{\Lambda}\Fbf^\top=\text{Toeplitz}[\tilde{s}_p]_{p=0,1,\ldots,n-1},
\end{equation}
where
\begin{equation}\label{eqn:tildesp}
\tilde{s}_p=\frac{1}{n}\sum_{j=0}^{n-1} w^{-pj} s_j
\end{equation}
\end{lemma}
\begin{proof}
The entry on the $l$-th row and the $m$-th column of $\Fbf \boldsymbol{\Lambda}\Fbf^\top$ is
\begin{equation}
\begin{split}
[\Fbf \boldsymbol{\Lambda}\Fbf^\top]_{l,m}=&\sum_{j=0}^{n-1} \frac{w^{lj}}{\sqrt{n}} \frac{w^{-mj}}{\sqrt{n}}s_j\\
=& \frac{1}{n}\sum_{j=0}^{n-1} w^{(l-m)j} s_j.
\end{split}
\end{equation}
Thus, Lemma~\ref{lmm:Fourier} holds.
\end{proof}
Therefore, the variance of all entries of $\Fbf\boldsymbol{\Lambda}\Fbf^\top$ is the same because
\begin{equation}\label{eqn:var_all}
\begin{split}
\text{var}_f[\tilde{s}_p]=&\text{var}_f\left[\frac{1}{n}\sum_{j=0}^{n-1} w^{-pj} s_j\right]\\
=&\frac{1}{n}\text{var}_f\left[s_0\right]=:\frac{1}{n}v.
\end{split}
\end{equation}
Further, the means of all diagonal entries of $\Fbf\boldsymbol{\Lambda}\Fbf^\top$ are
\begin{equation}\label{eqn:mean_diag}
\Ep_f[\tilde{s}_0]=\Ep_f\left[s_0\right]=:\mu,
\end{equation}
while the means of all off-diagonal entries are
\begin{equation}\label{eqn:mean_off_diag}
\Ep_f[\tilde{s}_p]=\frac{1}{n}\sum_{j=0}^{n-1} w^{-pj} \Ep_f[s_j]=0,\forall p\neq 0.
\end{equation}
\subsubsection{Using the concentration of \texorpdfstring{$\J_4$}{Lg} to obtain the error when \texorpdfstring{$n\to\infty$}{Lg}}
From an intuitive perspective, when $n\to\infty$, the submatrix $\J_4$ concentrates at $\mu \I_{n-k}$ (see the above computation on the mean and variance of all entries). In this case
\begin{equation}
\begin{split}
\Ep_f[\trace (\J_2\J_4^{-1}\J_2^\top)]\approx &\frac{1}{\mu} \Ep_f[\trace (\J_2\J_2^\top)]\\
=&\frac{1}{\mu}k(n-k)\text{var}[\tilde{s}_p]\\
=&\frac{n-k}{\mu}{v}\cdot\frac{k}{n}.
\end{split}
\end{equation}
Therefore, we have
\begin{equation}\label{eqn:final_ineq}
\begin{split}
&\lim_{n\to\infty}\frac{1}{n-k}\Ep_f[\trace(\J_2\J_4^{-1}\J_2^\top)]
=\frac{v}{\mu}= \frac{\text{var}_f[s_0]}{\Ep_f[s_0]}.
\end{split}
\end{equation}

Now, we formalize the above intuitive statement. In fact, we will show a even stronger bound than the bound on the expected error.
\begin{lemma}\label{lmm:concentration}
When $n-k=o(\sqrt{n})$, with high probability (in $1-\Obf(\frac{(n-k)^2}{n})$),
\begin{equation}
\begin{split}
&\frac{1}{n-k}\trace(\J_2\J_4^{-1}\J_2^\top)
\ge \frac{1}{\mu+\epsilon}\left(\frac{k}{n}v-\epsilon\right),
\end{split}
\end{equation}
for any $\epsilon>0$.
\end{lemma}
After we prove Lemma \ref{lmm:concentration}, we obtain a bound on expectation using the fact that
\begin{equation}
\begin{split}
    &\frac{1}{n-k}\Ep_f[\trace(\J_2\J_4^{-1}\J_2^\top)]\\
    &\ge (1-\Obf(\frac{(n-k)^2}{n}))\frac{1}{\mu+\epsilon}\left(\frac{k}{n}v-\epsilon\right).
\end{split}
\end{equation}
Thus, when $n\to\infty$ and $n-k=o(\sqrt{n})$, \begin{equation}\label{eqn:con_der1}
\begin{split}
&\lim_{n\to\infty}\frac{1}{n-k}\Ep_f[\trace(\J_2\J_4^{-1}\J_2^\top)]
\ge\frac{v-\epsilon}{\mu+\epsilon}= \frac{\text{var}_f[s_0]-\epsilon}{\Ep_f[s_0]+\epsilon},
\end{split}
\end{equation}
for all $\epsilon>0$, which completes the proof of Theorem~\ref{thm:beat_replication}.

The proof of Lemma~\ref{lmm:concentration} relies on the concentration of $\trace(\J_2\J_2^\top)$ and the concentration of $\J_4$. In particular, when we prove the concentration of $\J_4$, we use the Gershgorin circle theorem \cite{golub2012matrix}.
First, we show the following Lemma.

\begin{lemma}\label{lmm:Bsquare}
When $n-k=o(n)$, with high probability (in $1-\Obf(\frac{n-k}{n})$)
\begin{equation}\label{eqn:B2square}
\frac{1}{n-k}\trace(\J_2\J_2^\top)\ge\frac{k}{n}v-\epsilon.
\end{equation}
\end{lemma}

\begin{proof}
Since $(\J_2)_{k\times(n-k)}:=[\J_{i,j}]$ ($\J_{i,j}$ represents the entry on the $i$-th row and the $j$-th column) is the upper-right submatrix of $\Fbf \boldsymbol{\Lambda}\Fbf^\top=\text{Toeplitz}[\tilde{s}_p]_{p=0,1,\ldots,n-1}$,
\begin{equation}\label{eqn:B2sqaure2}
\begin{split}
\trace(\J_2\J_2^\top)=&\sum_{i=1}^k\sum_{j=1}^{n-k}|\J_{i,j}|^2\\
=&\sum_{l=1}^k\sum_{m=k+1}^n |\tilde{s}_{m-l}|^2.
\end{split}
\end{equation}
Since all entries in $\J_2$ have zero mean (because $l\neq m$ ever in \eqref{eqn:B2sqaure2} and from \eqref{eqn:mean_off_diag} all off-diagonal entries have zero mean) and have the same variance $\frac{v}{n}$ (see \eqref{eqn:var_all}),
\begin{equation}\label{eqn:B2sqaure3}
\begin{split}
\Ep_f\left[\frac{1}{n-k}\trace(\J_2\J_2^\top)\right]=&\frac{1}{n-k}\cdot k(n-k)\Ep_f[|\tilde{s}_1|^2]\\
\overset{(a)}{=}&\frac{1}{n-k}\cdot k(n-k)\var_f[\tilde{s}_1]=\frac{k}{n}v,
\end{split}
\end{equation}
where (a) holds because $\Ep_f[\tilde{s}_1]=0$. To prove \eqref{eqn:B2square}, we compute the variance of $\trace(\J_2\J_2^\top)$ and use Chebyshev's inequality to bound the tail probability. Define
\begin{equation}
\mu_B:=\Ep_f[\trace(\J_2\J_2^\top)]\overset{(a)}{=}\frac{k(n-k)}{n}v,
\end{equation}
where (a) follows from \eqref{eqn:B2sqaure3}. From \eqref{eqn:B2sqaure2}, we have
\begin{equation}\label{eqn:Parseval}
\begin{split}
\trace(\J_2\J_2^\top)\le& (n-k)\sum_{p=1}^{n-1}|\tilde{s}_p|^2\\
\overset{(a)}{=}&(n-k)\left(\frac{1}{n}\sum_{j=0}^{n-1} s_j^2-|\tilde{s}_0|^2\right),
\end{split}
\end{equation}
where the last equality (a) holds due to Parseval's equality for the Fourier transform, which states that $\frac{1}{n}\sum_{j=0}^{n-1} s_j^2=\sum_{p=0}^{n-1}|\tilde{s}_p|^2$.
Then,
\begin{equation}\label{eqn:B2sqaure4}
\begin{split}
&\var_f\left[\frac{1}{n-k}\trace(\J_2\J_2^\top)\right]\\
=&\Ep_f\left[\left(\frac{1}{n-k}\trace(\J_2\J_2^\top)\right)^2\right]-\Ep_f^2\left[\frac{1}{n-k}\trace(\J_2\J_2^\top)\right]\\
\overset{(a)}{\le} &\Ep_f\left[\left(\frac{1}{n}\sum_{j=0}^{n-1}s_j^2-|\tilde{s}_0|^2\right)^2\right]-\frac{k^2}{n^2}v^2\\
\overset{(b)}{=}&  \Ep_f\left[\left(\frac{1}{n}\sum_{j=0}^{n-1}s_j^2-(\frac{1}{n}\sum_{j=0}^{n-1}s_j)^2\right)^2\right]-\frac{k^2}{n^2}v^2,
\end{split}
\end{equation}
where (a) follows from \eqref{eqn:B2sqaure3} and \eqref{eqn:Parseval} and (b) follows from \eqref{eqn:tildesp}. Note that
\begin{equation}\label{eqn:B2sqaure5}
\frac{1}{n}\sum_{j=0}^{n-1}s_j^2-(\frac{1}{n}\sum_{j=0}^{n-1}s_j)^2=\frac{n-1}{n}s^2,
\end{equation}
where
\begin{equation}
s^2:=\frac{1}{n-1}\sum_{j=0}^{n-1}(s_j-\bar{s})^2,
\end{equation}
is the famous statistic called ``unbiased sample variance'', and its variance is (see Page 229, Theorem 2 in \cite{mood1950introduction})
\begin{equation}
\text{var}[s^2]=\frac{1}{n}\left(\mu_4-\frac{n-3}{n-1}\mu_2^2\right),
\end{equation}
where
\begin{equation}
\mu_4=\Ep[(s_0-\mu)^4],
\end{equation}
and
\begin{equation}
\mu_2=\Ep[(s_0-\mu)^2]=\text{var}[s_0]=v.
\end{equation}
Also note that the sample variance is unbiased, which means
\begin{equation}
\Ep_f[s^2]=v.
\end{equation}
Therefore, we have
\begin{equation}\label{eqn:B2sqaure6}
\Ep_f[(s^2)^2]=\text{var}[s^2]+(\Ep_f[s^2])^2=\frac{1}{n}\left(\mu_4-\frac{n-3}{n-1}v^2\right)+v^2,
\end{equation}
so we have
\begin{equation}\label{eqn:B2sqaure_var}
\begin{split}
&\var_f\left[\frac{1}{n-k}\trace(\J_2\J_2^\top)\right]\\
\overset{(a)}{\le}&  \Ep_f\left[\left(\frac{1}{n}\sum_{j=0}^{n-1}s_j^2-(\frac{1}{n}\sum_{j=0}^{n-1}s_j)^2\right)^2\right]-\frac{k^2}{n^2}v^2\\
\overset{(b)}{=}&\Ep_f\left[(\frac{n-1}{n}s^2)^2\right]-\frac{k^2}{n^2}v^2\\
=&\frac{(n-1)^2}{n^2}\Ep_f[(s^2)^2]-\frac{k^2}{n^2}v^2\\
\overset{(c)}{=}&\frac{(n-1)^2}{n^2}\frac{1}{n}\left(\mu_4-\frac{n-3}{n-1}v^2\right)+\frac{(n-1)^2-k^2}{n^2}v^2\\
=&\Obf\left(\frac{1}{n}\right)+\frac{(n-1)^2-k^2}{n^2}v^2,
\end{split}
\end{equation}
where (a) follows from \eqref{eqn:B2sqaure4}, (b) follows from \eqref{eqn:B2sqaure5} and (c) follow from \eqref{eqn:B2sqaure6}.

Note that we have computed the expectation of $\frac{1}{n-k}\trace(\J_2\J_2^\top)$, which is $\frac{k}{n}v$ (see \eqref{eqn:B2sqaure3}). Using the Chebyshev's inequality
\begin{equation}
\begin{split}
&\Pr\left(\left|\frac{1}{n-k}\trace(\J_2\J_2^\top)-\frac{k}{n}v\right|\ge \epsilon\right)\\
&\le \frac{1}{\epsilon^2} \text{var}\left[\frac{1}{n-k}\trace(\J_2\J_2^\top)\right]\\
&\overset{(a)}{\le}\frac{1}{\epsilon^2} \Obf\left(\frac{1}{n}\right)+\frac{1}{\epsilon^2}\frac{(n-1)^2-k^2}{n^2}v^2\\
&=\frac{1}{\epsilon^2}\Obf\left(\frac{1}{n}\right)+\frac{1}{\epsilon^2}\frac{(n-k-1)(n+k-1)}{n^2}v^2\\
&\overset{(b)}{<} \frac{1}{\epsilon^2}\Obf\left(\frac{1}{n}\right)+ \frac{2}{\epsilon^2}\frac{n-k-1}{n}v^2\\
&=\frac{1}{\epsilon^2}\Obf\left(\frac{n-k}{n}\right).
\end{split}
\end{equation}
where (a) is from \eqref{eqn:B2sqaure_var} and (b) is because $n+k-1<2n$. Therefore, the proof of \eqref{eqn:B2square} is over.
\end{proof}

Next, we show that with high probability the largest eigenvalue of $(\J_4)_{(n-k)\times(n-k)}$ is smaller than $(1+\epsilon)\mu$. Note that the matrix $\J_4$ is a principle submatrix of the Toeplitz matrix $\Fbf \boldsymbol{\Lambda}\Fbf^\top=\text{Toeplitz}[\tilde{s}_p]_{p=0,1,\ldots,n-1}$, so $\J_4=\text{Toeplitz}[\tilde{s}_p]_{p=0,1,\ldots,n-k-1}$ is also Toeplitz. Using the Gershgorin circle theorem, all eigenvalues of $\J_4:=[\tilde\J_{ij}]$ must lie in the union of $(n-k)$ circles, in which the $i$-th circle is centered at the diagonal entry $\tilde\J_{ii}=\tilde{s}_0$ and has radius $\sum_{j\neq i}|\tilde\J_{ij}|=\sum_{j\neq i}|\tilde{s}_{j-i}|$. These $(n-k)$ circles are all within the circle centered at $\tilde{s}_0$ with radius $2\sum_{p=1}^{n-k-1} |\tilde{s}_p|$. Therefore, the maximum eigenvalue of $\J_4$ satisfies
\begin{equation}
\sigma_{\text{max}}<\tilde{s}_0+2\sum_{p=1}^{n-k-1} |\tilde{s}_p|.
\end{equation}
Thus,
\begin{equation}\label{eqn:evalue_concentration}
\begin{split}
&\Pr(\sigma_{\text{max}}>\mu+\epsilon)<\Pr\left(\tilde{s}_0+2\sum_{p=1}^{n-k-1} |\tilde{s}_p|>\mu+\epsilon\right)\\
=&\Pr\left(\left(\tilde{s}_0-\mu+2\sum_{p=1}^{n-k-1} |\tilde{s}_p|\right)^2>\epsilon^2\right)\\
\overset{(a)}{\le}&\frac{1}{\epsilon^2}\Ep\left[\left(\tilde{s}_0-\mu+2\sum_{p=1}^{n-k-1} |\tilde{s}_p|\right)^2\right]\\
\overset{(b)}{\le} & \frac{1}{\epsilon^2} (2n-2k-1)^2\frac{v}{n}=\frac{1}{\epsilon^2}\Obf\left(\frac{(n-k)^2}{n}\right),
\end{split}
\end{equation}
where (a) is from the Markov inequality and (b) is due to the fact that $\var[\tilde{s}_p]=\frac{v}{n}$ for all $p$ and $\Ep[\tilde{s}_0]=\mu$ and $\Ep[\tilde{s}_p]=0$ for all $p\neq 0$.

From Lemma~\ref{lmm:Bsquare} and \eqref{eqn:evalue_concentration}, when $n\to\infty$ and $(n-k)^2=o(n)$, with high probability (which is $1-\frac{1}{\epsilon^2}\Obf\left(\frac{(n-k)^2}{n}\right)$),
\begin{equation}
\frac{1}{n-k}\trace(\J_2\J_2^\top)\ge\frac{k}{n}v-\epsilon,
\end{equation}
and at the same time
\begin{equation}
\J_4^{-1}\succeq \frac{1}{\mu+\epsilon}\I_{n-k}.
\end{equation}
From concentration of $\trace(\J_2\J_2^\top)$ and the lower bound of $\J_4^{-1}$, we have, with high probability,
\begin{equation}
\frac{1}{n-k}\trace(\J_2\J_4^{-1}\J_2^\top)
\ge \frac{1}{\mu+\epsilon}\left(\frac{k}{n}v-\epsilon\right),
\end{equation}
for all $\epsilon$. This concludes the proof of Lemma~\ref{lmm:concentration} and hence completes the proof of Theorem~\ref{thm:beat_replication} (see the details from after Lemma~\ref{lmm:concentration} to equation \eqref{eqn:con_der1}). This lemma is a formal statement of equality \eqref{eqn:final_ineq}.

%
\subsection{Proof of Theorem~\ref{thm:TtoInf}}\label{app:error_exponenl}

\textbf{1) Uncoded linear inverse problem:}

Consider the eigenvalue decomposition
\begin{equation}
\B=\P \Theta \P^{-1},
\end{equation}
where
\begin{equation}\label{eqn:gamma_der1}
\Theta=\diag\{\gamma_1,\gamma_2,\ldots \gamma_N\},
\end{equation}
and without the loss of generality, assume $\gamma_1$ is the maximum eigenvalue. Then, from the definition $\C(l_i)=\B^{l_i}\C_E(\B^\top)^{l_i}$ in \eqref{eqn:CTI},
\begin{equation}\label{eqn:cli_der1}
\C(l_i)=\P \Theta^{l_i} \P^{-1}\C_E(\P^\top)^{-1} \Theta^{l_i} \P^\top.
\end{equation}
Since $\P^{-1}\C_E(\P^\top)^{-1}$ is a positive definite matrix, all of its eigenvalues are positive real numbers. Suppose the maximum eigenvalue and the minimum eigenvalue of $\P^{-1}\C_E(\P^\top)^{-1}$ are respectively $e_{\max}$ and $e_{\min}$. Then, \eqref{eqn:cli_der1} gives the upper and lower bounds
\begin{equation}\label{eqn:cli_der21}
\trace(\C(l_i))\le e_{\max} \trace(\P\Theta^{2l_i}\P^\top),
\end{equation}
and
\begin{equation}\label{eqn:cli_der22}
\trace(\C(l_i))\ge e_{\min} \trace(\P\Theta^{2l_i}\P^\top).
\end{equation}
Suppose the maximum and minimum eigenvalues of $\P^\top\P$ are respectively $c_{\max}$ and $c_{\min}$. Then, \eqref{eqn:cli_der21} and \eqref{eqn:cli_der22} can be further simplified to
\begin{equation}\label{eqn:cli_der31}
\begin{split}
\trace(\C(l_i))&\le e_{\max}\trace(\Theta^{l_i}\P^\top \P\Theta^{l_i})\\
&\le c_{\max}e_{\max}\trace(\Theta^{2l_i})\\
&=c_{\max}e_{\max}\sum_{j=1}^N \gamma_j^{2l_i},
\end{split}
\end{equation}
and
\begin{equation}\label{eqn:cli_der32}
\begin{split}
\trace(\C(l_i))&\ge e_{\min}\trace(\Theta^{l_i}\P^\top \P\Theta^{l_i})\\
&\ge c_{\min}e_{\min}\trace(\Theta^{2l_i})\\
&=c_{\min}e_{\min}\sum_{j=1}^N \gamma_j^{2l_i},
\end{split}
\end{equation}
where the last equality in the above two inequalities are from the definition of $\Theta$ in \eqref{eqn:gamma_der1}. Therefore,
\begin{equation}
\begin{split}
&\lim_{T_\text{dl}\to\infty}\frac{1}{T_\text{dl}}\log \Ep[\twonorm{\Ebf_\text{uncoded}}^2|\mathbf{l}]\\
\overset{(a)}{=}& \lim_{T_\text{dl}\to\infty}\frac{1}{T_\text{dl}}\log \left(\sum_{i=1}^k \trace\left(\C(l_i)\right)\right)\\
\overset{(b)}{=}&\lim_{T_\text{dl}\to\infty}\frac{1}{T_\text{dl}}\log \left( \sum_{j=1}^{N} \sum_{i=1}^{k} \gamma_j^{2l_i}\right)\\
=&\lim_{T_\text{dl}\to\infty}\frac{1}{T_\text{dl}} \log \left( \sum_{j=1}^{N} \sum_{i=1}^{k} \gamma_j^{2\lceil \frac{T_\text{dl}}{v_i}\rceil}\right)\\
\overset{(c)}{=}&\max_{i\in[k],j} \log (\gamma_j)^{\frac{2}{v_i}},
\end{split}
\end{equation}
where (a) is from \eqref{eqn:error_uncoded}, (b) is obtained by plugging in \eqref{eqn:cli_der31} and \eqref{eqn:cli_der32} and the fact that the constants $e_\text{min}$, $c_\text{min}$, $e_\text{max}$ and $c_\text{min}$ do not change the error exponent when $T_{dl}$ increases, and (c) is from the fact that the maximum term dominates the error exponent in a log-sum form. Since the maximum eigenvalue of the matrix $\B$ is $\gamma_1$, we have
\begin{equation}
\begin{split}
\lim_{T_\text{dl}\to\infty}\frac{1}{T_\text{dl}}\log \Ep[\twonorm{\Ebf_\text{uncoded}}^2|\mathbf{l}]
=&\max_{i\in[k]} \log \gamma_1^{\frac{2}{v_i}}\\
=& -\frac{2}{\max_{i\in[k]} v_i} \log\frac{1}{\gamma_1}.
\end{split}
\end{equation}
Therefore, the error exponent is determined by the worker with the slowest speed (maximum $v_i$).

\textbf{2) replication-based linear inverse:}

Now we look at the replication-based linear inverse scheme. At first, we do not know the order of the random sequence $v_1,v_2,\ldots v_n$. Therefore, when we assign the extra $n-k<k$ workers to replicate the computations of the last $n-k$ linear inverse problems, there is a non-zero probability that the slowest worker of the first $k$  workers does not have any other copy. More precisely, denote by $E$ the above event. Then, if we uniformly choose $n-k$ workers to replicate, the probability of $E$ is
\begin{equation}
\Pr(E)=\frac{\binom{k-1}{n-k}}{\binom{k}{n-k}}.
\end{equation}
This is also a constant that does not depend on the time $T_\text{dl}$. Therefore,
\begin{equation}
\Ep[\twonorm{\Ebf_\text{rep}}^2|\mathbf{l}]\ge c_{\min}e_{\min}\Pr(E) \sum_{j=1}^N \gamma_j^{2\lceil\frac{T_\text{dl}}{\max_{i\in[k]} v_i}\rceil},
\end{equation}
where the exponent ${2\lceil\frac{T_\text{dl}}{\max_{i\in[k]} v_i}\rceil}$ is because we are lower-bounding the error of replication-based scheme using only the error of the slowest worker in the first $k$ workers, and $\Pr(E)$ is the probability that this particular worker is not replicated using any of the $n-k$ extra workers.

Using the fact that $\max_j \gamma_j=\gamma_1$ and the fact that $c_{\min}e_{\min}\Pr(E)$ is a constant that does not change with $T_{dl}$, we have
\begin{equation}
\lim_{T_\text{dl}\to\infty}\frac{1}{T_\text{dl}}\log\Ep[\twonorm{\Ebf_\text{rep}}^2|\mathbf{l}]\ge \frac{2}{\max_{i\in [k]} v_i} \log\frac{1}{\gamma_1}.
\end{equation}
Note that $\Ep[\twonorm{\Ebf_\text{rep}}^2|\mathbf{l}]\le \Ep[\twonorm{\Ebf_\text{uncoded}}^2|\mathbf{l}]$, so we also have
\begin{equation}
\begin{split}
&\lim_{T_\text{dl}\to\infty}\frac{1}{T_\text{dl}}\log\Ep[\twonorm{\Ebf_\text{rep}}^2|\mathbf{l}]\\
&\le \lim_{T_\text{dl}\to\infty}\frac{1}{T_\text{dl}}\log\Ep[\twonorm{\Ebf_\text{uncoded}}^2|\mathbf{l}]\\
&=\frac{2}{\max_{i\in [k]} v_i} \log\frac{1}{\gamma_1}.
\end{split}
\end{equation}
Therefore,
\begin{equation}
\lim_{T_\text{dl}\to\infty}\frac{1}{T_\text{dl}}\log\Ep[\twonorm{\Ebf_\text{rep}}^2|\mathbf{l}]= \frac{2}{\max_{i\in [k]} v_i} \log\frac{1}{\gamma_1}.
\end{equation}

\textbf{3) Coded linear inverse algorithm:}

For the coded linear inverse algorithm,
\begin{equation}\label{eqn:000}
\Ep[\twonorm{\Ebf_\text{coded}}^2|\mathbf{l}]\le\sigma_\text{max}(\Gbf^\top\Gbf)\trace\left[(\Gbf \boldsymbol{\Lambda}^{-1} \Gbf^\top)^{-1}\right].
\end{equation}
From \eqref{eqn:cli_der31}, we have
\begin{equation}
\begin{split}
\trace(\C(l_i))\le& c_{\max}e_{\max}\sum_{j=1}^N \gamma_j^{2l_i}\\
\le&c_{\max}e_{\max} N \gamma_1^{2l_i}.
\end{split}
\end{equation}
Plugging into \eqref{eqn:000}, we have
\begin{equation}
\Ep[\twonorm{\Ebf_\text{coded}}^2|\mathbf{l}]\le\sigma_\text{max}(\Gbf^\top\Gbf)\trace\left[(\Gbf \boldsymbol{\Lambda}_2^{-1} \Gbf^\top)^{-1}\right],
\end{equation}
where
\begin{equation}
\begin{split}
\boldsymbol{\Lambda}_2:=&\diag \{c_{\max}e_{\max} N \gamma_1^{2l_1},\ldots c_{\max}e_{\max} N \gamma_1^{2l_n}\}\\
=&N c_{\max}e_{\max} \diag \{\gamma_1^{2l_1},\ldots \gamma_1^{2l_n}\}.
\end{split}
\end{equation}
Since $N$, $c_{\max}$ and $e_{\max}$ are all constant numbers,
\begin{equation}\label{eqn:006}
\begin{split}
&\lim_{T_\text{dl}\to\infty}\frac{1}{T_\text{dl}}\log \Ep[\twonorm{\Ebf_\text{coded}}^2|\mathbf{l}]\\
\le& \lim_{T_\text{dl}\to\infty}\frac{1}{T_\text{dl}} \log \trace\left[(\Gbf \boldsymbol{\Lambda}_3^{-1} \Gbf^\top)^{-1}\right],
\end{split}
\end{equation}
where
\begin{equation}\label{eqn:lambda3}
\begin{split}
\boldsymbol{\Lambda}_3
:=&\diag \{\gamma_1^{2l_1},\ldots \gamma_1^{2l_n}\}\\
=&\diag \{\gamma_1^{2\lceil\frac{T_\text{dl}}{v_1}\rceil},\ldots \gamma_1^{2\lceil\frac{T_\text{dl}}{v_n}\rceil}\}.
\end{split}
\end{equation}
Define $\S=\{i_1,\ldots i_k\}$, i.e., the index set of the fastest $k$ workers. Then,
\begin{equation}\label{eqn:001}
\min_{i\in \S} \left(\frac{1}{\gamma_1}\right)^{2\lceil \frac{T_\text{dl}}{v_i}\rceil}=\left(\frac{1}{\gamma_1}\right)^{2\lceil\frac{T_\text{dl}}{v_{i_k}}\rceil}.
\end{equation}
For $i\in [n]\setminus \S=\{i_{k+1},\ldots i_n\}$,
\begin{equation}\label{eqn:002}
\left(\frac{1}{\gamma_1}\right)^{2\lceil \frac{T_\text{dl}}{v_i}\rceil}\ge 0.
\end{equation}
Therefore, from the definition of the diagonal matrix $\boldsymbol{\Lambda}_3$ in \eqref{eqn:lambda3}, the entries of $\boldsymbol{\Lambda}_3^{-1}$ can be lower-bounded by \eqref{eqn:001} for $i\in\S$, and can be lower-bounded by \eqref{eqn:002} for $i\in [n]\setminus \S$. Thus,
\begin{equation}\label{eqn:003}
\boldsymbol{\Lambda}_3^{-1}\succeq \left(\frac{1}{\gamma_1}\right)^{2\lceil\frac{T_\text{dl}}{v_{i_k}}\rceil} \diag\{c_1,c_2,\ldots c_n\},
\end{equation}
where $c_i$ is the indicator
\begin{equation}\label{eqn:004}
c_i=\delta(i\in\S).
\end{equation}
Define $\Gbf_{\T}$ as the submatrix of $\Gbf$ composed of the columns in $\Gbf$ with indexes in $\T\subset [n]$. Use $\sigma_{\min}(\X)$ to denote the minimum eigenvalue of a matrix $\X$. Define
\begin{equation}\label{eqn:005}
s_{\min}=\min_{\T\subset [n],|\T|=k}\sigma_{\min}(\Gbf_{\T}\Gbf_{\T}^\top).
\end{equation}
Since $\Gbf$ is a matrix with orthonormal rows, any arbitrary $\Gbf_{\T}$ that satisfies $|\T|=k$ must have full rank. This means that $s_{\min}>0$. Note that $s_{\min}>0$ is a constant that depends only on the generator matrix $\Gbf$ and does not change with the overall time $T_\text{dl}$. Therefore,
\begin{equation}\label{eqn:007}
\begin{split}
\Gbf\boldsymbol{\Lambda}_3^{-1}\Gbf^\top&
\overset{(a)}{\succeq}\left(\frac{1}{\gamma_1}\right)^{2\lceil\frac{T_\text{dl}}{v_{i_k}}\rceil} \Gbf\diag\{c_1,c_2,\ldots c_n\}\Gbf^\top\\ &\overset{(b)}{=}\left(\frac{1}{\gamma_1}\right)^{2\lceil\frac{T_\text{dl}}{v_{i_k}}\rceil} \Gbf_{\S}\Gbf_{\S}^\top\\
&\overset{(c)}{\succeq}\left(\frac{1}{\gamma_1}\right)^{2\lceil\frac{T_\text{dl}}{v_{i_k}}\rceil} s_{\min}\I_k.
\end{split}
\end{equation}
where (a) is from \eqref{eqn:003}, (b) is from \eqref{eqn:004}, and (c) is from \eqref{eqn:005}. Thus, plugging \eqref{eqn:007} into \eqref{eqn:006} (note that there is an inverse inside the trace of \eqref{eqn:006})
\begin{equation}
\begin{split}
&\lim_{T_\text{dl}\to\infty}\frac{1}{T_\text{dl}}\log \Ep[\twonorm{\Ebf_\text{coded}}^2|\mathbf{l}]\\
\le &\lim_{T_\text{dl}\to\infty}\frac{1}{T_\text{dl}} \log \trace\left[(\Gbf \boldsymbol{\Lambda}_3^{-1} \Gbf^\top)^{-1}\right]\\
\le &\lim_{T_\text{dl}\to\infty}\frac{1}{T_\text{dl}} \log \trace \left[\left(\frac{1}{\gamma_1}\right)^{-2\lceil\frac{T_\text{dl}}{v_{i_k}}\rceil} \frac{1}{s_{\min}}\I_k\right]\\
= &\lim_{T_\text{dl}\to\infty}\frac{1}{T_\text{dl}} \log \left\{\left(\frac{1}{\gamma_1}\right)^{-2\lceil\frac{T_\text{dl}}{v_{i_k}}\rceil}\trace \left[ \frac{1}{s_{\min}}\I_k\right]\right\}\\
\overset{(a)}{=}&\lim_{T_\text{dl}\to\infty}\frac{1}{T_\text{dl}} \log \left(\frac{1}{\gamma_1}\right)^{-2\lceil\frac{T_\text{dl}}{v_{i_k}}\rceil}\\
=&-\frac{2}{v_{i_k}} \log\frac{1}{\gamma_1},
\end{split}
\end{equation}
where (a) is because $\trace \left[ \frac{1}{s_{\min}}\I_k\right]=\frac{k}{s_{\min}}$ is a constant and does not change the error exponent. Thus, we have completed the proof of Theorem~\ref{thm:TtoInf}.

\subsection{Computing the Matrix \texorpdfstring{$\boldsymbol{\Lambda}$}{Lg}}\label{app:computing_Lambda}
Recall that the statistic $\hat{\gamma}_{m,l}$ is defined as
\begin{equation}
\hat{\gamma}_{m,l}=\frac{1}{m}\sum_{j=1}^m \twonorm{\B^l\a_j}^2,l=1,2,\ldots T_u.
\end{equation}
The computational complexity of computing $\hat{\gamma}_{m,l},l=1,2,\ldots T_u$ is the same as the computation of $m$ linear inverse problems for $T_u$ iterations. The computation has low complexity and can be carried out distributedly in $m$ workers before the main algorithm starts. Additionally, the computation results can be used repeatedly when we implement the coded linear inverse algorithm multiple times. \textcolor{black}{In the data experiments, we use $m=10$, which has the same complexity as solving $m=10$ extra linear inverse problems.}

The following Lemma shows that $\hat{\gamma}_{m,l},l=1,2,\ldots T_u$ is an unbiased and asymptotically consistent estimate of $\trace(\C(l))$ for all $l$.

\begin{lemma}\label{lmm:trace_estimate}
The statistic $\hat{\gamma}_{m,l}$ is an unbiased and asymptotically consistent estimator of $\trace(\C(l))$. More specifically, the mean and variance of the estimator $\hat{\gamma}_{m,l}$ satisfies
\begin{equation}
\Ep[\hat{\gamma}_{m,l}|\mathbf{l}]=\trace(\C(l)),
\end{equation}
\begin{equation}
\var_t[\hat{\gamma}_{m,l}]\le\frac{1}{m}\twonorm{\B^l}_F^4\Ep\left[\twonorm{\a_j}^4\right].
\end{equation}
\end{lemma}
\begin{proof}
The expectation of $\hat{\gamma}_{m,l}$ satisfies
\begin{equation}
\begin{split}
\Ep[\hat{\gamma}_{m,l}]=&\frac{1}{m}\sum_{j=1}^m \Ep\left[\twonorm{\B^l\a_j}^2\right]\\
=&\Ep\left[\twonorm{\B^l\a_1}^2\right]\\
=&\Ep\left[\trace(\B^l\a_1\a_1^\top (\B^l)^\top)\right]\\
\overset{(a)}{=}&\trace(\B^l\Ep[\a_1\a_1^\top] (\B^l)^\top)\\
=&\trace(\B^l\C_E(\B^l)^\top)\\
=&\trace(\C(l)),
\end{split}
\end{equation}
where (a) is from the fact that $\a_1$ has covariance $\C_E$. To bound the variance of $\hat{\gamma}_{m,l}$, note that for all $j$,
\begin{equation}
\begin{split}
\twonorm{\B^l\a_j}^2\le &\twonorm{\B^l}_F^2\twonorm{\a_j}^2.
\end{split}
\end{equation}
Therefore,
\begin{equation}
\begin{split}
\var[\hat{\gamma}_{m,l}]=&\var[\frac{1}{m}\sum_{j=1}^m \twonorm{\B^l\a_j}^2]\\
\overset{(a)}{=}&\frac{1}{m}\var\left[\twonorm{\B^l\a_j}^2\right]\\
\overset{(b)}{\le}&\frac{1}{m}\Ep\left[\twonorm{\B^l\a_j}^4\right]\\
\overset{(c)}{\le}&\frac{1}{m}\twonorm{\B^l}_F^4\Ep\left[\twonorm{\a_j}^4\right],
\end{split}
\end{equation}
\end{proof}
where (a) holds because all $\twonorm{\a_j}$ are independent of each other, and (b) holds because $\var[X]\le \Ep[X^2]$, and (c) is from the Cauchy-Schwartz inequality.
\subsection{Proof of Theorem~\ref{thm:complexity}}\label{app:complexity}
The computational complexity at each worker is equal to the number of operations in one iteration multiplied by the number of iterations. The number of iterations is $l$. In each iteration, the number of operations is equal to the number of non-zeros in $\B$ because each iteration $\x^{(l+1)}=\Kbf \rbf+\mathbf{B}\x^{(l)}$ requires at least scanning through the non-zeros in $\B$ once to compute $\B\x^{(l)}$. For general matrices, the number of entries is in the order of $N^2$, where $N$ is the number of rows in $\B$. Therefore, the overall number of operations at each worker is in the order of $\Theta(N^2l)$.

The encoding and decoding steps in Algorithm~\ref{alg:pg} are all based on matrix-matrix multiplications. More specifically, for encoding, we multiply the generator matrix $\Gbf_{k\times n}$ with the input matrix and the initial estimates, which both have size $N\times k$. Thus, the complexity scales as $\Obf(nkN)$. For decoding, the computation of the decoding matrix $\Lbf=(\Gbf^\top\boldsymbol{\Lambda}^{-1}\Gbf)^{-1}\Gbf\boldsymbol{\Lambda}^{-1}$ is has complexity $\Theta(k^3)$ (matrix inverse) plus $\Theta(k^2n)$ (matrix multiplications). Multiplying the decoding matrix $\Lbf_{k\times n}$ with linear inverse results that have size $N\times n$ has complexity $\Theta(nkN)$. Therefore, for large $N$, the computational complexity is in the order of $\Theta(nkN)$.

The computation of the matrix $\boldsymbol{\Lambda}$, as we have explained in Section~\ref{sec:complexity}, has the same complexity as computing $m\approx 10$ extra linear inverse problems. Additionally, it is a one-time cost in the pre-processing step. Thus, we do not take into account the complexity of computing $\boldsymbol{\Lambda}$ for the analysis of encoding and decoding.

\subsection{Proof of Theorem~\ref{thm:communication_cost}}\label{app:communication_cost_proof}

We assume that the matrix $\B$ and $\Kbf$ have already been stored in each worker before the computation of the linear inverse problems. For the PageRank problem, this means that we store the column-normalized adjacency matrix $\A$ in each worker.

In Algorithm~\ref{alg:pg}, the $i$-th worker requires the central controller to communicate a vector $\rbf_i$ with length $N$ to compute the linear inverse problem. Thus
\begin{equation}
\text{COST}_\text{communication}= N\quad \text{INTEGERS}.
\end{equation}

The computation cost at each worker is equal to the number of operations in one iteration multiplied by the number of iterations in the specified iterative algorithm. In each iteration, the number of operations also roughly equals to the number of non-zeros in $\B$. Thus
\begin{equation}
\text{COST}_\text{computation}\approx 2\cdot |\E|\cdot l_i \text{  OPERATIONS},
\end{equation}
where $l_i$ is the number of iterations completed at the $i$-th worker, $|\E|$ is the number of non-zeros, and $2\cdot$ is because we count both addition and multiplication. From Fig. \ref{fig:various_speed}, the typical number of $l_i$ is about 50.

Thus, the ratio between computation and communication is
\begin{equation}\begin{split}
&\text{COST}_\text{computation}/\text{COST}_\text{communication}\\
&\approx l_i \bar{d}\text{  OPERATIONS/INTEGERS},
\end{split}
\end{equation}
where $\bar{d}$ is the average number of non-zeros in each row of the $\B$ matrix. Since $l_i$ is about 50, we expect that the computation cost is much larger than communication.

\subsection{Proof of Lemma~\ref{lmm_vec1}}\label{app:vec}
Property 1 and property 2 can be directly examined from the definition. Property 3 is Theorem~3 in \cite{zhang2013kronecker}. To prove property 4, we note that the eigenvalues of $\A\otimes\B$ equals to the pairwise products of the eigenvalues of $\A$ and the eigenvalues of $\B$ (from Theorem 6 in \cite{zhang2013kronecker}). Therefore, since the eigenvalues of $\A$ and the eigenvalues of $\B$ are all non-negative, the eigenvalues of $\A\otimes\B$ are also non-negative. Property 5 follows directly from property 4 because when $\B-\A\succeq \0$ and $\C\succeq \0$, $(\B-\A)\otimes \C\succeq \0$.

To prove property 6, we can repeatedly use property 3:
\begin{equation}
\begin{split}
&(\A_{m\times n}\otimes \I_p)\cdot(\I_n\otimes\B_{p\times q})\\
=&(\A_{m\times n}\cdot \I_n)\otimes (\I_p\cdot\B_{p\times q})\\
=&(\I_m\cdot\A_{m\times n})\otimes (\B_{p\times q}\cdot \I_q)\\
=&(\I_m\otimes \B_{p\times q})\cdot (\A_{m\times n}\otimes \I_q).
\end{split}
\end{equation}

To prove property 7, we first assume that
\begin{equation}
\Lbf_{k\times n}=\left[\begin{matrix}
L_{11}&L_{12}&\ldots&L_{1n}\\
L_{21}&L_{22}&\ldots&L_{2n}\\
\vdots&\vdots&\ddots&\vdots\\
L_{n1}&L_{n2}&\ldots&L_{kn}
\end{matrix}\right].
\end{equation}
Then,
\begin{equation}
\begin{split}
&\trace\left[(\Lbf\otimes \I_N)\cdot \A\cdot (\Lbf\otimes \I_N)^\top\right]\\
\overset{(a)}{=}&\sum_{l=1}^k\trace\left[\sum_{i=1}^n\sum_{j=1}^nL_{ki}\A_{ij}L_{kj}\right]\\
=&\sum_{l=1}^k\left[\sum_{i=1}^n\sum_{j=1}^nL_{ki}\trace[\A_{ij}]L_{kj}\right]\\
\overset{(b)}{=}&\trace\left[\Lbf\cdot\left[\begin{matrix}
\trace[\A_{11}]&\ldots &\trace[\A_{1n}]\\
\vdots & \ddots & \vdots\\
\trace[\A_{n1}]&\ldots &\trace[\A_{nn}]
\end{matrix}\right]\cdot\Lbf^\top\right],
\end{split}
\end{equation}
where (a) and (b) hold both because the trace can be computed by examining the trace on the diagonal (or the diagonal blocks).

\bibliographystyle{ieeetr}
\bibliography{rough}

\end{document}